\documentclass[journal,twoside,web]{IEEEtran}

\usepackage{cite}
\usepackage{amsmath,amssymb,amsfonts, amsthm}
 \usepackage{algorithm}
\usepackage{algorithmicx}
\usepackage{algpseudocode}
\usepackage{graphicx}
\usepackage{subcaption}
\usepackage{cases}
\usepackage{textcomp}
\usepackage{bbm}
\usepackage{url}
\PassOptionsToPackage{hyphens}{url}\usepackage{hyperref}

\newtheorem{theorem}{\bf Theorem}

\newtheorem{remark}{\bf Remark}

\newtheorem{proposition}{\bf Proposition}
\newtheorem{lemma}{\bf Lemma}

\newcommand{\indep}{\perp \!\!\! \perp}

\newcommand{\bx}{{\rvec{ x}}}
\newcommand{\by}{{\rvec{y} }}

\newcommand{\bw}{{ \rvec{w}}}
\newcommand{\bv}{{\rvec{ v}}}
\newcommand{\ba}{{\rvec{a}}}
\newcommand{\floor}[1]{\lfloor {#1} \rfloor}

\newcommand{\rvec}[1]{\boldsymbol{\mathsf{#1}} }
\newcommand{\dvec}[1]{{\mathsf{#1}} }
\newcommand{\dmat}[1]{{\mathsf{#1}} }
\newcommand{\Tr}{\text{\upshape{Tr}}}
\newcommand{\st}{\text{\upshape{s.t. }}}

\newcommand{\Cesaro}{Ces\'{a}ro }

\begin{document}
\title{Minimum Bitrate Neuromorphic Encoding for Continuous-Time  Gauss-Markov Processes}
\author{Travis Cuvelier, Ronald Ogden, and Takashi Tanaka 
\thanks{T. Cuvelier is with the Chandra Department of Electrical and Computer Engineering at the University of Texas at Austin,  Austin, TX 78712 USA (e-mail: tcuvelier@utexas.edu). T. Tanaka, and R. Ogden, are with the Department of Aerospace Engineering and Engineering Mechanics at the University of Texas at Austin, Austin, TX 78712 USA (e-mail: ronnieogden95@utexas.edu, ttanaka@utexas.edu). }
}

\maketitle

\begin{abstract}
In this work, we study minimum data rate tracking of a dynamical system under a neuromorphic event-based sensing paradigm. We begin by bridging the gap between continuous-time (CT) system dynamics and information theory's causal rate distortion theory. We motivate the use of non-singular source codes to quantify bitrates in event-based sampling schemes. This permits an analysis of minimum bitrate event-based tracking using tools already established in the control and information theory literature. We derive  novel, nontrivial lower bounds to event-based sensing, and compare the lower bound with the performance of well-known schemes in the established literature. 
\end{abstract}


\section{Introduction}\label{sec:introduction}
Neuromorphic, or ``event-based", sensing is a biologically inspired paradigm that is fundamentally distinct from both conventional, periodic digital sensing and ``continuous-time" approaches. Rather than producing samples at discrete instances, an event-based sensor produces a stream of event tuples $(g(t_0),t_0),(g(t_1),t_1),(g(t_2),t_2),...$, where the sampling instances $t_i$ are chosen by the sensor real-time. A salient feature of event-based sampling is that samples are produced ``as needed", according to an application-specific rule. In an event-based framework, information about the source process $g(t)$ is encoded both in the samples $g(t_{i})$ and in their timing, $t_i$. In scenarios where the samples must be encoded and transmitted over a network, an event-based paradigm allows the encoder to remain ``silent" for long intervals.  In some scenarios, neuromorphic sampling can reduce the use of the communication/computational resources \cite{astrom2002comparison, guo2021optimal_IT}. Relatedly, there has been a recent resurgence of interest in event-based vision. Event-based cameras use asynchronous pixel readout mechanisms that are fundamentally different from their conventional, frame-based counterparts. They offer low latency and high-dynamic range with reduced power consumption \cite{censi2015power}\cite{prophesee2020neurips}. There is recent work that explores the application of event-cameras in computer vision \cite{gallego2020event} and as sensors for guidance, navigation and control of robotics \cite{mueggler2017event}. 

A foundational problem in sensing, vision, and control over communication networks is to minimize the use of communication resources subject to a constraint on quality. For any fixed reliability, data rates can be directly tied to the amount of scarce physical resources (time/bandwidth/power) that must be allocated to achieve the desired performance \cite{hjjournal}. The rate-distortion trade-off delineates the minimum bitrate required to achieve some desired control or estimation performance. Understanding this trade-off is key to the development of sensing and communication protocols; at a minimum, it sets the goalposts for system optimization, and at best, gives insight into optimal protocol designs. While in control and real-time vision, the system of interest usually operates in CT (e.g. an underlying visual scene), and the distortion metric depends on real-time reconstruction at some decoder (e.g. the pixels on an LCD screen), rate-distortion tradeoffs are typically formulated in terms of a discrete-time (DT) model that assumes a fixed sampling period.  These DT formulations are not amenable to the analysis of more general event-based sensing and communication modalities, and may obscure more realistic notions of communication and estimation performance as they arise in real time. With notable exceptions (cf. \cite{guo2021optimal_TAC,guo2021optimal_IT}), there is little work on event-based estimation and control in the modern literature on rate-distortion theory. 

Meanwhile, on the experimental side, event-based sensing paradigms have shown promise as a means for communication-efficient control over wireless networks. While it has been proposed that radios required for large-scale industrial automation must achieve extremely high reliability and low latency \cite{dahlman20145g}\cite{durisi2016toward}, using event-based control/communication co-design has demonstrated success with cheaper, less reliable radios \cite{baumann2021wireless}, motivating a rate-distortion-theoretic investigation of event-based sensing. 

The experimental success of event-based paradigms for communication efficient control and tracking motivates the pursuit of a theoretical understanding of the tradeoff between bitrate and control/estimation performance in systems that use event-based sampling. In this work, we formulate a causal rate-distortion optimization that captures an accurate notion of the bitrate/estimation performance tradeoff when a CT Gauss-Markov source process is observed, encoded, and reconstructed using an event-based sensing and communication paradigm. We assume that an encoder continuously monitors the source. The encoder produces discrete messages, nominally binary strings, and asynchronously conveys them over a noiseless channel to an estimator. The estimator then produces a causal CT estimate of the source, and its performance is quantified by its mean squared error (MSE). Our optimization aims to minimize the bitrate of the channel subject to a constraint on this MSE.  We subsequently combine tools from the theories of dynamical systems, estimation/control, and source coding to derive novel, nontrivial lower bounds on this rate-distortion tradeoff. To our knowledge, this is the first such result for vector-valued  Gauss-Markov source processes. We then provide a detailed analysis of the derived bound, and compare it to the performance achievable by classical event-based tracking algorithms. 

\subsection{Related work}
A variety of event-based paradigms for sensing and control have been studied over several decades. An overview of some classical approaches to event-based sampling in networked control systems (NCS) is available in \cite{heemels2012introduction}.  In the NCS context, we say an ``event is triggered" when a particular control or sensing action is taken.  

Most of the work in event-based network control assumes a real number communication model that neglects quantization. Under this model, there has been significant work on the design of event-based sampling rules that account for ``communication cost" as quantified by communication frequency. In a significant early work on linear-quadradic Gaussian (LQG) control,  \cite{astrom2002comparison} demonstrated that a simple, ad-hoc, thresholding policy to trigger plant measurements coupled with impulsive control could outperform a comparable, periodically sampled, system for the same expected number of plant measurements per unit time. The gain of the CT approach was shown to be nearly achievable for sufficiently high-rate sampled systems. For Gauss-Markov tracking, \cite{marck2010relevant} devised an informational sampling rule that provably ensures the estimator's error covariance remains bounded. In   \cite{imer2010optimal}, an observer makes $n$ DT noisy measurements of a source process and, given a fixed horizon, causally selects a subset $m<n$ of the measurements to convey to an estimator who seeks to minimize some distortion metric. In a similar setting,  \cite{lipsa2011remote} proposes to jointly optimize a weighted average of real-time communication frequency and mean square estimator error. A problem formulation falling between that of  \cite{ imer2010optimal} and \cite{lipsa2011remote} is treated in \cite{wu2012event}. Considering tracking a Wiener process in CT, \cite{narsampling2014} established that the sampling rule that minimizes the expected number of samples for a fixed estimator MSE is a threshold policy that samples when the real-time estimator error departs an elliptical region. Along similar lines, \cite{sun2020sampling} considered a setup where an observer conveys event-triggered real-valued samples of a Wiener process to a remote estimator over a queueing channel with random delay. Given a sample rate constraint and an estimator conditioned (only) on the sequence of causally received samples, the optimal sampling policy is shown to be a prescribed threshold policy \cite{sun2020sampling}. In contrast with this line of work, the event-based sensing and communication paradigm we propose accounts for quantization. We assume that the messages conveyed from the encoder to the estimation center are finite-length binary strings, and quantify communication cost via the strings' length. 

The literature on event-based control with discrete \textit{quantized} measurements is most relevant to this work.  An event-based sampling scheme, coupled with a quantizer design, that stabilizes a Gauss-Markov plant with bounded disturbances and communication delay is proposed in \cite{tallapragada2015event}. Fixed-length, dynamic quantization is used; in other words, a fixed number of bits are transmitted to the controller at every event-triggered sampling time \cite{tallapragada2015event}. Also relevant is \cite{pearson2017control}, which uses fixed-length quantization to stabilize a deterministic Gauss-Markov plant. In \cite{pearson2017control} the communication, rather than the sampling, is event-based. With sampling and symbol transmission times preset,  one particular symbol is designated as ``free". In other words, \cite{pearson2017control} allows for the possibility that a ``lack of transmission" or "silence" can convey information. A necessary and sufficient condition for stabilization is derived in terms of the system dynamics, the product of the quantizer size and sampling rate (quantizer bitrate), and the fraction of non-free symbols transmitted (utilization). An event-based encoder is proposed and is shown to stabilize the system with a quantizer bitrate and utilization on the same order as the optimum. A salient feature of event-based encoding is that information is conveyed in the content of transmissions as well as \textit{when} transmissions occur. This is explored in detail in \cite{khojasteh2020value}, which studies the stabilization of deterministic Gauss-Markov plants under a continuous-time ad hoc event-based sampling policy coupled with variable-length feedback. The communication channel is assumed to have a bounded, but unknown communication delay. It is demonstrated that if the delay is sufficiently short, the timing information conveyed by the event triggers is sufficient to guarantee stabilization \cite{khojasteh2020value}. Longer maximum delays require increasing the packet bitrate \cite{khojasteh2020value}. The recent work \cite{guo2021optimal_TAC, guo2021optimal_IT} is most relevant to our present investigation. In the more-general \cite{guo2021optimal_IT}, an encoder continuously monitors a scalar-valued Markov source and conveys variable-length bit strings asynchronously to a decoder that produces a CT reconstruction of the source process. Communication cost is quantified in the expected length of these strings per unit time.  For the class of sources considered, \cite{guo2021optimal_IT} devises an optimal encoder policy that minimizes the reconstruction MSE for a constraint on the expected bitrate. The optimal policy consists of a threshold rule similar to \cite{astrom2002comparison} that conveys one-bit messages to the decoder. In most of this work, we consider CT tracking of scalar stable Gauss–Markov processes, a subset of the class of sources considered in \cite{guo2021optimal_IT}, and use a nominally identical notion of communication cost. In contrast to \cite{guo2021optimal_IT}, while the lower bounds we derive apply to event-based schemes with a minimum temporal sampling resolution, they can be directly extended to vector sources. 

This work uses tools from minimum bitrate Gauss-Markov tracking and linear-quadratic Gaussian (LQG) control. DT, minimum bitrate causal tracking of vector Gauss-Markov sources is considered in \cite{tanaka2016semidefinite,photisSRDF_STSP}. Assuming a noiseless binary channel from a source observer to an estimation center, \cite{tanaka2016semidefinite} and \cite{photisSRDF_STSP} derive rate-distortion trade-offs that characterize the minimum expected bitrate of prefix-free coding required to achieve a given infinite horizon estimator performance. Similarly,  \cite{ kostina2019rate,tanaka2017lqg} used an analogous formulation to characterize the tradeoff between bitrate and LQG control performance. In contrast to this prior work, we consider CT source and reconstruction processes, and propose a CT notion of communication cost. In particular, we use ideas from \cite{szpankowski2011minimum,kostina2019rate} to account for a relaxation of prefix constraints. 

\subsection{Our Contributions} Our contributions are three-fold:  
\begin{enumerate}
\item For CT, vector-valued, time-invariant Gauss-Markov processes with a fixed sampling period $\tau$, we derive a lower bound to the data rate achievable by real-time codecs utilizing a non-prefix binary codeword every $\tau$ seconds to satisfy a mean-square distortion constraint. The bound incorporates ideas from \textit{minimum information Kalman-Bucy filtering} \cite{tanaka2022gaincontrol}. Given the solution of a single semidefinite optimization that depends only on the continuous-time state space model and target distortion, our bound can be computed immediately for any $\tau$. 

\item The CT bound we derive follows from (initially) casting the CT problem in a completely DT framework. The DT analysis provides another, tighter bound. Computing this bound, however, is more difficult; it requires the solution of a semidefinite optimization for every combination of state-space model, distortion target, and $\tau$. 

\item We investigate these bounds numerically. While we find that our bound becomes vacuous as $\tau\rightarrow 0$, it is significantly tighter for shorter sampling periods. Our simulation results also provide one of the first characterizations of the bitrate/mean square distortion tradeoff for the causal tracking of vector-valued Gauss-Markov plants. 
\end{enumerate}
Both of the bounds we develop are applicable to general real-time encoding scenarios for CT vector-valued Gauss-Markov processes. They serve as a novel, universal benchmarks against which algorithms for communication-efficient event-based tracking may be compared.
\subsection{Notation}
Bold symbols, e.g. $\bx$ indicate random variables. We use serif fonts $x$ ($\bx$) for (random) scalars,  sans-serif lower case $\dvec{x}$ ($\rvec{x}$) for (random) vectors, sans-serif capitals for matrices $\dmat{A}$. For $\dmat{A}$, $\dmat{B}$ symmetric, $\dmat{A}\succ \dmat{B}$ ($\dmat{A}\succeq \dmat{B}$) implies $\dmat{A}-\dmat{B}$ is positive (semi)definite. For $\dmat{Q}\succ 0$, $\lVert \dvec{v}\rVert^{2}_{\dmat{Q}} = \dvec{v}^{\top}\dmat{Q}\dvec{v}$, and $\lVert \dvec{v}\rVert^{2} =  \dvec{v}^{\top}\dmat{Q}\dvec{v}$. We denote $\dmat{v}$'s transpose as $\dvec{v}^{\top}$.  We use $\indep$ to denote independence,  $I(\bx;\by)$ the mutual information between $\bx$ and $\by$, and for $\ba$ a discrete random variable, $H(\ba)$ its entropy. CT (random) signals at time $t\in\mathbb{R}$ are denoted with subscripts $x_{t}$ ($\bx_{t}$). For $I\subset\mathbb{R}$ and $\mathbbm{1}_{t\in I}$ the indicator of the set $I$, $\bx_{I}$ denotes the random process $\mathbbm{1}_{t\in I}\bx_{t}$. We also consider periodic samples of CT (random) signals $x_{t}$ $(\bx_{t})$ with the sampling period $\tau$. For $k\in\mathbb{N}_{0}$, let $x(k)=x_{k\tau}$ (resp. $\bx(k) = \bx_{k\tau}$). For $i,j\in\mathbb{N}_{0}$, let $x(i:j) = (x_{i},x_{i+1},\dots,x_{j})$ if $j\ge i$ and $x(i:j)=\emptyset$ otherwise. Denote the set of finite-length binary strings $\{0,1\}^{*}$. If $a\in\{0,1\}^*$, $\ell(a)$ denotes the length of $a$. 

\section{System Model and Problem Formulation}
We consider a communication architecture shown in Fig.~\ref{fig:event_based_encoding}. 
\begin{figure*}[h]
    \centering   \includegraphics[width=0.9\textwidth]{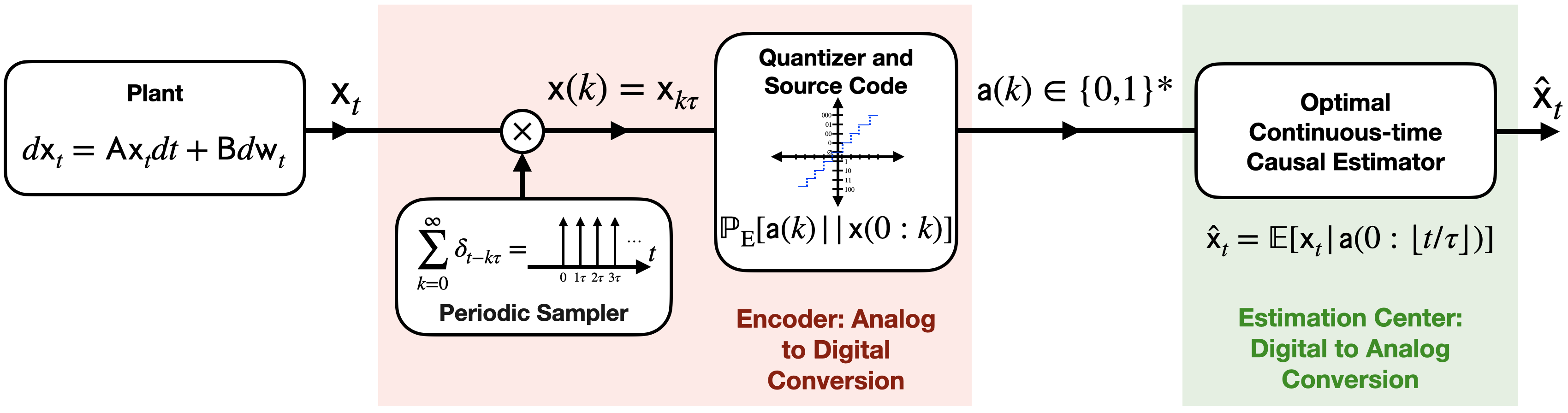}
    \caption{ We consider a scenario in which a CT source signal is sampled, quantized, and encoded into finite-length binary packets. We assume that an event-based encoder samples the source signal, then quantizes and then encodes the samples into packets containing finite-length binary strings. We assume the packets are conveyed without delay over an error-free binary channel to a remote ``estimation center". The estimation center uses the packets it has received by time $t$ to construct a CT estimate of the source process.}
    \label{fig:event_based_encoding}
    \vspace{-.5cm}
\end{figure*}
The signal to be tracked is an $n$-dimensional CT random process $\bx_t, t\geq0$ defined by a linear stochastic differential equation (SDE):
\begin{equation}
\label{eq:source_c}
d\rvec{x}_t = \dmat{A}\rvec{x}_t dt+\dmat{B}d\rvec{w}_t, \; \rvec{x}_0\sim \mathcal{N}(0,\dmat{\Sigma}_{0}), \; t \geq 0,
\end{equation} where  $\dmat{A}\in\mathbb{R}^{n\times n}$, $\dmat{B}\in\mathbb{R}^{n\times n}$, and $\dmat{\Sigma}_{0}\succ 0$ are known and $\rvec{w}_t\in\mathbb{R}^{m}$ is a standard Brownian motion. We assume $\dmat{A}\in\mathbb{R}^{m\times m}$ is Hurwitz stable (i.e. the eigenvalues of $\dmat{A}$ have strictly negative real part), and that $\dmat{B}\dmat{B}^{\top}\succ 0$. The encoder observes the process \eqref{eq:source_c} at discrete time instances with a fixed sampling period $\tau >0$ observing the sequence $\bx_\tau, \bx_{2\tau}, \bx_{3\tau}, \cdots$ noiselessly. For $k\in\mathbb{N}_{0}$, we denote $x(k) = \bx_{k\tau}$. At each $t=k\tau$, the encoder can convey a variable length packet of bits, denoted $\rvec{a}(k)\in \{\emptyset, 0, 1, 00, 01,\dots \}$ to the decoder, or estimation center. The length of these packets will be used to quantify communication cost; by permitting the encoder to transmit the empty string  $\emptyset$, we allow for the possibility of ``no transmission". The encoder's packets are drawn via a causal encoding policy, nominally the sequence of Borel-measurable stochastic kernels 
\begin{align}\label{eq:encoderspace}
   \mathbb{P}_{\mathrm{E}}[ \ba(k)|| \bx(k)] =\{ \mathbb{P}\left[ \ba(k)| \ba(0:k-1), \bx(0:k)\right]\}_{k\in\mathbb{N}_{0}}.
\end{align} In other words, each packet is chosen given the causally available samples and previous packets. We assume the packets are transmitted without error and without delay. While in a truly event-based sampling paradigm, the encoder continuously monitors the source process (\ref{eq:source_c}) and can transmit messages to the encoder at arbitrary instants, real-world systems have finite bandwidth. We use $\tau$ to constrain the minimum temporal resolution of the sensing and/or communication hardware. This is natural; any channel of (baseband) bandwidth $W$ can support at most $2W$ distortion-free transmissions per second \cite{nyquist}. We will consider $\tau>0$ as a variable, and will recover a CT analysis by letting $\tau\rightarrow 0$. 

At every $t$, the estimation center produces the \textit{CT} estimate $\hat{\bx}_t$ using \textit{only} the packets it has received on or before time $t$ (i.e. $\ba(0:\floor{t/\tau})$). We quantify the estimation center's tracking error via the time average MSE
\begin{align}\label{eq:mse}
     \limsup_{T\rightarrow \infty} \frac{1}{T}\int_0^T \mathbb{E}[\lVert\bx_t-\hat{\bx}_t\rVert_2^2] dt. 
\end{align} Communication cost is quantified by the time-averaged expected length of the codewords $\{\ba(k)\}$. The infinite horizon cost has the units ``bits per second" and, is given by
\begin{align}
    \limsup_{T\rightarrow\infty } \frac{1}{T}\sum_{k=0}^{\floor{\frac{T}{\tau}}}  \mathbb{E}[\ell(\ba(k))] =  \limsup_{K\rightarrow \infty} \frac{1}{K\tau}\sum_{k=0}^{K-1} \mathbb{E}[\ell(\ba(k))].\nonumber
\end{align} This communication cost is to be contrasted with the more common ``bits per sample" notion seen in the literature on minimum bitrate tracking for DT systems (c.f., e.g., \cite{tanaka2016semidefinite,photisSRDF_STSP,tanaka2017lqg,kostina2019rate}). Our goal is to design encoding and estimation policies that minimize the time-average expected bitrate (in bits/second) of the channel such that the causal estimate $\hat{\bx}_t$ satisfies a constraint on the distortion (\ref{eq:mse}). Since for any encoding policy, choosing $\hat{\bx}_{t} = \mathbb{E}[\bx_{t}|\ba(0:\floor{t/\tau})]$  minimizes (\ref{eq:mse}), we assume without loss of generality that this is the decoder's estimate. For $D_{c}>0$ the constraint on estimator distortion, this leads to the following optimization over encoding policies \begin{align} \label{eq:initialOptimization}
\mathcal{L}^*_{\mathrm{CT}}(D_{c},\tau) =  
\left\{
\begin{aligned}
&\inf_{\mathbb{P}_{\mathrm{E}}} \limsup_{K\rightarrow \infty} \frac{1}{K\tau}\sum_{k=0}^{K-1}  \mathbb{E}[\ell(\ba(k))] \\ &\st \limsup_{T\rightarrow \infty} \frac{1}{T}\int_0^T \mathbb{E}[\Vert\bx_t-\hat{\bx}_t\rVert_2^2] dt \leq D_c
\end{aligned}
\right. \end{align} 
Given the minimum sampling period $\tau$, it may seem this optimization can be reduced to those treated in prior work on causal rate-distortion theory (e.g. \cite{tanaka2016semidefinite, photisSRDF_STSP,kostina2019rate}). In fact, this formulation differs from conventional treatments in several ways. This permits us to retain many of the essential characteristics of event-driven sampling paradigms. Since by convention, the empty string has length $\ell(\emptyset)=0$, the encoder can choose not only which packet to send, but \textit{when} to send it; this is notionally equivalent to an event-driven encoder choosing not to transmit. Additionally, (\ref{eq:initialOptimization}) differs from \cite{tanaka2016semidefinite,photisSRDF_STSP,kostina2019rate} in that the distortion constraint  in (\ref{eq:initialOptimization}) is in continuous time. It turns out, however, that the properties of the sampled Ito processes allow us to immediately recover a more conventional purely DT formulation, which will be used to derive our first lower bound. A second, relaxed lower bound will follow from a return to a completely CT viewpoint.

\section{Main results}\label{sec:main}
In this section, we give a high-level overview of our main result, which is a lower bound on (\ref{eq:initialOptimization}). This will require solving a tractable convex optimization that involves the system matrices. Given the solution to this optimization, one can derive a lower bound on (\ref{eq:initialOptimization}) for any $\tau$. The derivation of this result is relegated to Section \ref{sec:proofs}.

The lower bound is based on CT {minimum information Kalman-Bucy filtering}. Minimum information Kalman-Bucy filtering, introduced in \cite{tanaka2022gaincontrol}, can be seen as a CT generalization 
of classical discrete-time minimum bitrate causal tracking (cf. e.g. \cite{tanaka2016semidefinite,photisSRDF_STSP,kostina2019rate}). Suppose that the process \eqref{eq:source_c} is observed through an $n$-dimensional observation process
\begin{equation}
\label{eq:obs}
d\by_t=\dmat{C}_{t} \bx_t dt+d\bv_t, 
\end{equation}
where $\dmat{C}_{t}$ is a time-varying sensor gain matrix (formally $\dmat{C}_{t}:\mathbb{R}_{+}\rightarrow \mathbb{R}^{n\times n}$), and $\bv_t$ is an $n$-dimensional standard Brownian motion independent of $\bw_t$. We assume that the CT measurement sequence $\rvec{y}_{t}$ is conveyed to a decoder, which computes the MMSE estimator via Kalman-Bucy filtering. Explicitly, the decoder computes
 $\hat{\bx}_t=\mathbb{E}[\bx_t|\by_s, 0\leq s \leq t]$ via the standard Kalman-Bucy filter
\begin{equation}
\label{eq:kb_filter}
d\hat{\bx}_t=\dmat{A}\hat{\bx}_t dt + \dmat{X}_{t} \dmat{C}_{t}^\top(d\by_t-\dmat{C}_{t} \hat{\bx}_t dt), \;\; \hat{\bx}_{0}=0,
\end{equation} where the error covariance matrix is denoted  $\dmat{X}_{t} = \mathbb{E}\left[ (\bx_t-\hat{\bx}_t)(\bx_t-\hat{\bx}_t)^{\top}\right]$. Rather than assuming that $\dmat{C}_{t}$ is given,  we will consider choosing $\dmat{C}_{t}$ by optimization. For any choice of sensor gain trajectory $\dmat{C}_{t}$, we have that $\dmat{X}_{t}$ satisfies the Riccati boundary value problem
\begin{equation}\label{eq:riccati}
\dot{\dmat{X}}_t=\dmat{A}\dmat{X}_{t}+\dmat{X}_{t}\dmat{A}^\top-\dmat{X}_{t}\dmat{C}_{t}^\top \dmat{C}_{t} \dmat{X}_{t}+\dmat{B}\dmat{B}^\top, \;\; \dmat{X}_0=\dmat{\Sigma}_{0}.
\end{equation} Both the CT squared error (\ref{eq:secost}) and  mutual information (\ref{eq:micost}) costs can be written in terms of $\dmat{X}_{t}$. We have \begin{subequations}\label{eq:riccatsimp}
\begin{align}
\int_0^T \mathbb{E} \|\bx_t-\hat{\bx}_t\|^2 dt &= \int_0^T   \Tr(\dmat{X}_{t}) dt\text{, and} \label{eq:mseX}\\
I(\bx_{[0,T]};\hat{\bx}_{[0,T]})&=\frac{1}{\ln(2) 2}\int_0^T   \Tr(\dmat{C}_{t} \dmat{X}_{t} \dmat{C}_{t}^\top) dt, \label{eq:mi}
\end{align} 
\end{subequations} where (\ref{eq:mi}) follows from \cite{duncan1970calculation} and is expressed in bits. The optimization proposed in \cite{tanaka2022gaincontrol} aims to choose $\dmat{C}_{t}$ to trade-off between the CT mutual information (\ref{eq:mi}) and squared error (\ref{eq:mseX}). We consider optimizing the time-varying gain $\dmat{C}_{t}$ over a space of functions. Let $\mathcal{F}_{\mathbb{R}\rightarrow \mathbb{R}^{m\times m}}$ be the space of measurable functions from $\mathbb{R}\rightarrow \mathbb{R}^{m\times m}$ such that if  $\dmat{C}_{t}\in\mathcal{F}_{\mathbb{R}\rightarrow \mathbb{R}^{m\times m}}$, then there exists  an absolutely continuous function $\dmat{X}_{t}$ with $\dmat{X}_{0}=\Sigma_{0}$ such that both $\lVert \dmat{X}_{t}\rVert_{2}<\infty$ and  $\dmat{X}_{t}\succ 0$ for all $t\in\mathbb{R}^{+}$ with  $\dmat{X}_{t}$ satisfying the differential equation in (\ref{eq:riccati}) almost everywhere. In other words so long as 
 $\dmat{C}_{t}\in\mathcal{F}_{\mathbb{R}\rightarrow \mathbb{R}^{m\times m}}$ then the solution to (\ref{eq:riccati}) exists in the Cartheradory sense for all $t\in\mathbb{R}^{+}$, and furthermore has $\dmat{X}_{t}\succ 0$ for all $t$. Consider the optimization 
\begin{subequations}\label{eq:rate_info_c}
\begin{align}
    \mathcal{I}^c(D) = \inf_{\dmat{C}_{t}\in \mathcal{F}_{\mathbb{R}\rightarrow \mathbb{R}^{m\times m}}} \quad & \limsup_{T\rightarrow +\infty} \frac{1}{T} I(\bx_{[0,T]};\hat{\bx}_{[0,T]})\label{eq:micost} \\
    \st \quad & \limsup_{T\rightarrow +\infty} \frac{1}{T} \int_0^T \mathbb{E}\|\bx_t-\hat{\bx}_t\|^2 dt \leq D \label{eq:secost} 
\end{align}
\end{subequations}
We call $\mathcal{I}^c(D)$ the \emph{CT information-distortion function}, and we will use it to develop a further lower bound on (\ref{eq:initialOptimization}). While the set $\mathcal{F}_{\mathbb{R}\rightarrow \mathbb{R}^{m\times m}}$ is difficult to characterize explicitly. It is, however, relatively rich; it includes, for example, any trajectory of $\dmat{C}_{t}$ such that 
$\dmat{C}_{t}^{\top}\dmat{C}_{t}$ locally integrable \cite{potterReport}.  Furthermore, it is not presently clear how to go about computing (\ref{eq:rate_info_c}), which is an infinite-dimensional optimization. It turns out that the constraint $\dmat{C}_{t}\in \mathcal{F}_{\mathbb{R}\rightarrow \mathbb{R}^{m\times m}}$ makes the optimization (\ref{eq:rate_info_c}) amenable to convexification, and in fact, given that the system dynamics in (\ref{eq:source_c}) are time-invariant, it allows  (\ref{eq:rate_info_c}) to be computed by a semidefinite program. This is the subject of the following lemma. 
\begin{lemma}\label{lemm:ct2sdp}
We have that \begin{subequations}\label{eq:the_ct_ti_sdp}
\begin{align}
{\mathcal{I}}^c(D)  = \inf_{\dmat{X} \succeq 0, \dmat{Y} \succeq 0} \quad &  \frac{a+   \Tr\left(\dmat{Y}\right)}{2\ln(2)}\label{eq:traceobjective4} \\
    \st & 0  \preceq \dmat{A}\dmat{X}+\dmat{X}\dmat{A}^\top+\dmat{B}\dmat{B}^\top,\\& \begin{bmatrix}
    \dmat{Y} & \dmat{B}^{\top} \\ \dmat{B} & \dmat{X}
    \end{bmatrix}\succeq 0  \\
    &   \Tr(\dmat{X}) \le D.\label{eq:matrixse4}
\end{align}
\end{subequations}
\end{lemma} Lemma \ref{lemm:ct2sdp} is proven in Section \ref{subsec:convexproofs}. The main result of this section is the following theorem, which is a lower bound on (\ref{eq:initialOptimization}) in terms of (\ref{eq:rate_info_c}). 
\begin{theorem}\label{thm:keybounds_abbrev}
    We have the following
    \begin{IEEEeqnarray}{rCl}
          \mathcal{L}^*_{\mathrm{CT}}(D_{c},\tau) &\ge& \frac{1}{\tau}\theta^{-1}\left(\tau \mathcal{I}^c(D_{c})\right)\label{eq:uselemmacont}
    \end{IEEEeqnarray}
\end{theorem} Theorem \ref{thm:keybounds_abbrev} is proven in Section \ref{sec:proofs}. To prove Theorem \ref{thm:keybounds_abbrev}, we will first show that (\ref{eq:initialOptimization}) is amenable to a DT formulation reminiscent of prior work on minimum rate causal tracking \cite{tanaka2016semidefinite}. This leads to a ``DT" lower bound on (\ref{eq:initialOptimization}) that is computed via a semidefinite program. The program is parameterized by discrete-time system matrices corresponding to (\ref{eq:source_c}) sampled with period $\tau$. We then return to a completely CT viewpoint to further relax this bound. The utility of the CT relaxation is that it requires only solving a single semidefinite optimization to bound the required bitrate for the entire range of feasible sampling periods. We compare the CT and DT bounds in Section \ref{sec:numerical}, and find close agreement for a wide range of $\tau$.  However, it should be noted that the intermediate result Theorem \ref{thm:firstbound} in Section \ref{subsec:dtbounds} provides a tighter bound for any fixed $\tau$. 
\section{Derivation of the main result }\label{sec:proofs}
As discussed in Section \ref{sec:main}, Theorem \ref{thm:keybounds_abbrev} follows from a relaxation of a discrete-time lower bound on (\ref{eq:initialOptimization}). We derive this bound in the next subsection. 
\subsection{DT rate-distortion lower bounds}\label{subsec:dtbounds}
The first lower bound follows via an analysis of the source process (\ref{eq:source_c}) sampled at a frequency of $1/\tau$. We first describe this sampled process, and our description will lead to a structural result that casts the optimization (\ref{eq:initialOptimization}) in a DT framework. With this framework in hand, we derive our first main result, which is a lower bound on (\ref{eq:initialOptimization}) derived via the solution of a DT causal rate-distortion optimization.

By the definition of the process (\ref{eq:source_c}), for $s\ge 0$
\begin{equation}
\label{eq:source_d_vec}
\rvec{x}_{t+s} = \dmat{A}_{s}\rvec{x}_{t}+\dmat{B}_{s}\rvec{w}_{t}
\end{equation}  where $\rvec{w}_{t}\sim\mathcal{N}(0,\dmat{I}_{m\times m})$ is a standard Gaussian $m$-dimensional random vector with $\rvec{w}_{t}\indep (\rvec{x}_{t},\ba(0:\floor{t/\tau}))$ and 
\begin{equation}\label{eq:ab_dt_vec}
\dmat{A}_{s}=e^{\dmat{A}s} \text{ and } \dmat{B}_{s}=\left(\int_0^{s} e^{\dmat{A}\lambda}\dmat{B}\dmat{B}^\top e^{\dmat{A}^\top\lambda}d\lambda\right)^{\frac{1}{2}}.
\end{equation} By the Hurwitz stability of $\dmat{A}$ that the limits $\lim_{t\rightarrow\infty}\mathbb{E}[\rvec{x}_{t}\rvec{x}_{t}^{\top}]= \lim_{t\rightarrow \infty} \int_0^{t} e^{\dmat{A}\lambda}\dmat{B}\dmat{B}^\top e^{\dmat{A}^\top\lambda}d\lambda$ are well defined and equal to the unique $\dmat{X}$ that satisfies the Lyaponov equation $\dmat{A}\dmat{X}+\dmat{X}\dmat{A}^\top+\dmat{B}\dmat{B}^\top =0$. Note that if $D_{c} >   \Tr(\dmat{X})$, the distortion constraint in (\ref{eq:initialOptimization}) is trivially satisfied by a policy where the encoder chooses $\rvec{a}(k)=\emptyset$ for all $k$ and the decoder uses the estimate $\rvec{\hat{x}}_{t}=0$.  Thus, we will restrict our attention to the more interesting case when $D_c <   \Tr(\dmat{X})$. 

Let $\rvec{x}(k) = \rvec{x}_{k\tau}$ and $\hat{\rvec{x}}(k) = \hat{\rvec{x}}_{k\tau}$ be the samples of the source process (\ref{eq:source_c}) and the estimator.  In particular, we have 
 \begin{align}\label{eq:discretizeddynamics_vec}
     \rvec{x}(k+1) =\dmat{A}_{\tau}\rvec{x}(k)+\dmat{B}_{\tau}\rvec{w}(k)
 \end{align} where the $\rvec{w}(k)\sim \mathcal{N}(0_{m},\dmat{I}_{m})$ and $\rvec{w}(k)\indep (\rvec{w}(0:k-1),\rvec{x}(0:k))$. Furthermore, since $\hat{\bx}_{t} = \mathbb{E}[\bx_{t}|\ba(0:\floor{t/\tau})]$ and (\ref{eq:source_d_vec}),
    $\hat{\rvec{x}}_{t} = \dmat{A}_{t-\tau\floor{t/\tau}}\hat{\rvec{x}}(\floor{t/\tau})$.
 Thus, the optimal CT reconstruction depends only on the reconstructions \textit{of the samples}, e.g. $\{\hat{\bx}(k)\}$. We interpret this as ``optimal interpolation". This allows the distortion constraint in (\ref{eq:mse}) to be cast in terms of the sampled reconstruction error.  
Let $\rvec{e}_{t} = \rvec{x}_t-\hat{\rvec{x}}_t$ and $\rvec{e}(k) = \rvec{x}_{k\tau}-\hat{\rvec{x}}_{k\tau}$. Assume $t\in [k\tau,(k+1)\tau)$ so that  $\floor{t/\tau}=k$.
We have, at every $t$
\begin{multline}\label{eq:mmsesimp}
   \mathbb{E}[\Vert\rvec{e}_{t}\rVert_{2}^{2}]=\mathbb{E}[\rvec{e}(k)^{\top}\dmat{A}_{t-k\tau }^{\top}\dmat{A}_{t-k\tau }\rvec{e}(k))] +\\  \Tr(\dmat{B}_{t-k\tau }\dmat{B}_{t-k\tau }^{\top}).
\end{multline} We will now derive an expression for the integral of (\ref{eq:mmsesimp}) over the interval $t\in [k\tau,(k+1)\tau)$. Let
\begin{IEEEeqnarray}{rCl}\label{eq:aintegrals}
\overline{\dmat{Q}}_{\dmat{A},\tau} = \int_{k\tau} ^{(k+1)\tau}\dmat{A}_{t-k\tau }^{\top}\dmat{A}_{t-k\tau }dt=
\int_{0}^{\tau} e^{\dmat{A}^{\top}t}e^{\dmat{A}t}dt,\label{eq:avl}
\end{IEEEeqnarray}  and let 
$\overline{b}_{\tau} = 
\Tr(\int_{k\tau}^{(k+1)\tau}\dmat{B}_{t-k\tau }\dmat{B}_{t-k\tau }^{\top}dt)$.  Both $\overline{\dmat{Q}}_{\dmat{A},\tau}$ and $\overline{b}_{\tau}$ depend on the sampling interval $\tau$, but not on $k$. Both are well-defined and, since $\dmat{A}$ is Hurwitz,  the integrals can be computed via \cite{vanloan}. Since  $\overline{\dmat{Q}}_{\dmat{A},\tau}\succ 0_{m\times m}$, we can define the norm ($\dvec{v}\in\mathbb{R}^{m}$) $\lVert \dvec{v}\rVert_{\overline{\dmat{Q}}_{\dmat{A},\tau}} = \sqrt{\dvec{v}^\top\overline{\dmat{Q}}_{\dmat{A},\tau}\dvec{v}}$. Applying the linearity of the integral and trace, as well as the  Fubini theorem, gives
\begin{align}\label{eq:finalSimplified}
     \int_{k\tau} ^{(k+1)\tau} \mathbb{E}[\Vert\rvec{x}_t-\hat{\rvec{x}}_t\rVert_{2}^{2}]dt = \mathbb{E}[\lVert \rvec{e}(k) \rVert^{2}_{\overline{\dmat{Q}}_{\dmat{A},\tau}}]+\overline{b}_{\tau}. 
\end{align} The following lemma allows us to cast the distortion constraint (\ref{eq:mse}) in terms of the $\rvec{e}(k)$, i.e. reconstruction error at the sampling times $t=k\tau$. \begin{lemma}\label{lemm:dimplication_vec}
Define 
\begin{align}\label{eq:DcDd_vec}
    D_{d,\tau } = D_{c}\tau-\overline{b_{\tau}}.
\end{align} If $\tau>0$ is sufficiently small, then $D_{d,\tau } >0$. If $D_{d,\tau } >0$, $\limsup_{T\rightarrow \infty} \frac{1}{T}\int_0^T \mathbb{E}[\lVert\rvec{x}_t-\hat{\rvec{x}}_t\rVert_2^2 ]dt \leq D_c$ then $\limsup_{K\rightarrow \infty} \frac{1}{K}\sum_{k=0}^{K-1} \mathbb{E}[\lVert \rvec{x}(k) - \hat{\rvec{x}}(k) \rVert_{{\overline{\dmat{Q}}_{\dmat{A},\tau}}}^2] \le   D_{d,\tau }$.
\end{lemma}
\begin{proof}[Proof of Lemma \ref{lemm:dimplication_vec}:]
By l'H\^{o}pital's Rule and the Fundamental Theorem of Calculus, $\lim_{\tau\rightarrow 0^{+}} (D_{c}\tau-\overline{b_{\tau}})/\tau =  D_{c}$.
 Thus for sufficiently small $\tau>0$, we have $D_{d,\tau }>0$. Furthermore, if $\limsup_{T\rightarrow \infty} \frac{1}{T}\int_0^T \mathbb{E}[\lVert\rvec{x}_t-\hat{\rvec{x}}_t\rVert_2^2 ]dt  \le D_{c}$ then for $K\in\mathbb{N}$
\begin{align}\label{eq:integralprebreakup}
     \limsup_{K\rightarrow \infty} \frac{1}{K\tau}\int_0^{K\tau} \mathbb{E}[\lVert\rvec{x}_t-\hat{\rvec{x}}_t\rVert_2^2 ]dt \le D_{c}.
\end{align} 
Breaking up (\ref{eq:integralprebreakup}) into integrals over $[k\tau,(k+1)\tau)$ and applying (\ref{eq:finalSimplified}) gives
\begin{eqnarray} 
    \lefteqn{
    \limsup_{K\rightarrow \infty} \frac{1}{K\tau}\int_0^{K\tau} \mathbb{E}[\lVert\rvec{x}_t-\hat{\rvec{x}}_t\rVert_2^2 ]dt} \nonumber \\ 
    &=& \limsup_{K\rightarrow \infty} \frac{1}{K\tau}\sum_{k=0}^{K}\int_{k\tau}^{(k+1)\tau} \mathbb{E}[\lVert\rvec{x}_t-\hat{\rvec{x}}_t\rVert_2^2 ]dt \\
    &=& \limsup_{K\rightarrow \infty} \frac{1}{K}\sum_{k=0}^{K} \frac{\mathbb{E}[\lVert \rvec{e}(k) \rVert^{2}_{\overline{\dmat{Q}}_{\dmat{A},\tau}}]}{\tau}+\frac{\overline{b}_{\tau}}{\tau}, \label{eq:l1vecpf}
\end{eqnarray} 
establishing the inclusion in the statement of the lemma. 
\end{proof}
Lemma \ref{lemm:dimplication_vec} and the preceding comments immediately lead to the following result, which allows us to derive a lower bound to (\ref{eq:initialOptimization}) under a DT framework.
\begin{proposition}\label{prop:cd2dtprop}
Denoting $\rvec{e}(k)=  \bx(k) - \hat{\bx}(k)$, for $D>0$, $Q\in \mathbb{R}^{m\times m}$, $\dmat{Q}\succeq 0$, define 
\begin{align}\label{eq:optim_first_simp}
    \mathcal{L}^*_{\mathrm{DT}}(D,\dmat{Q},\tau) = \left\{\begin{aligned} &\inf_{\mathbb{P}_{\mathrm{E}}} \limsup_{K\rightarrow \infty} \frac{1}{K\tau}\sum_{k=0}^{K-1} \mathbb{E}[\ell(\ba(k))] \\ &\st \limsup_{K\rightarrow \infty} \frac{1}{K}\sum_{k=0}^{K-1} \mathbb{E}[\lVert \rvec{e}(k)\rVert_\dmat{Q}^2] \le D \end{aligned} \right.,
\end{align} where $\hat{\bx}(k) = \mathbb{E}[\bx(k)|\ba(0:k)]$ for all $k$ and expectations and information measures are with respect to the joint measure induced by the DT system dynamics (\ref{eq:discretizeddynamics_vec}) and the choice of the encoding kernel $\mathbb{P}_{\mathrm{E}}$. For a fixed  $D_{c}>0$ with a sufficiently small $\tau$  such  that $D_{d,\tau}>0$ we have 
\begin{align}
\label{eq:ct_to_dt_distortion_bound}
    \mathcal{L}^*_{\mathrm{DT}}(D_{d,\tau},\overline{\dmat{Q}}_{\dmat{A},\tau},\tau) \le \mathcal{L}^*_{\mathrm{CT}}(D_{c},\tau).
\end{align}
\end{proposition}
\begin{remark}
In fact, it can be shown that the converse holds in Lemma \ref{lemm:dimplication_vec}, i.e. given the sampled data model, it follows that if $D_{d,\tau } >0$, $\limsup_{K\rightarrow \infty} \frac{1}{K}\sum_{k=0}^{K-1} \mathbb{E}[\lVert \rvec{x}(k) - \hat{\rvec{x}}(k) \rVert_{{\overline{\dmat{Q}}_{\dmat{A},\tau}}}^2] \le   D_{d,\tau }$ implies that $\limsup_{T\rightarrow \infty} \frac{1}{T}\int_0^T \mathbb{E}[\lVert\rvec{x}_t-\hat{\rvec{x}}_t\rVert_2^2 ]dt \leq D_c$. This implies that (\ref{eq:ct_to_dt_distortion_bound}) is actually an equality. Establishing the inequality (\ref{eq:ct_to_dt_distortion_bound}) is sufficient for our work here. 
\end{remark}

While the optimization in (\ref{eq:optim_first_simp}) resembles those studied causal rate-distortion literature (cf. e.g. \cite{tanaka2016semidefinite}, \cite{photisSRDF_STSP}), the absence of prefix-constraints, critical to the ``event-based" formulation, requires additional analysis. The main result of this section combines prior work on causal rate-distortion with tools from lossless compression without prefix constraints \cite{szpankowski2011minimum} to derive a lower bound on $\mathcal{L}^*_{\mathrm{DT}}(D_{d,\tau},\overline{\dmat{Q}}_{\dmat{A},\tau},\tau)$.  Let $\ba$ be a random variable with a range in $\{0,1\}^{*}$ (by convention $\emptyset\in\{0,1\}^{*}$).  Recall the definition of $\theta:\mathbb{R}^{+}\rightarrow\mathbb{R}^{+}$ from Theorem \ref{thm:firstbound}. Since $\theta$ is strictly increasing and concave, $\theta^{-1}$ is strictly increasing and convex.  A direct consequence of \cite[Section II]{szpankowski2011minimum} is that
\begin{align}\label{eq:noprefix}
    \theta^{-1}(H(\ba))  \le \mathbb{E}[\ell(\ba)].
\end{align} Combining (\ref{eq:noprefix}) with the fact that $\theta^{-1}$ is strictly increasing and convex leads to the following lemma. \begin{lemma}\label{lemm:mutual_information_lb} For any causal encoding policy $\mathbb{P}_{\mathrm{E}}$
\begin{multline}\label{eq:converse}
        \frac{1}{K\tau}\sum_{k=0}^{K-1}  \mathbb{E}[\ell(\ba(k))] \ge\\ \frac{1}{\tau}  \theta^{-1}(\frac{1}{K}I(\bx(0:K-1);\hat{\bx}(0:K-1))).
\end{multline}
\end{lemma} 

\begin{proof}
For any $K\in\mathbb{N}_{+}$, applying (\ref{eq:noprefix}) to every codeword $\ba(k)$ gives 
\begin{align}\label{eq:the_lb_we_relax}
    \frac{1}{K\tau}\sum_{k=0}^{K-1}  \mathbb{E}[\ell(\ba(k))] \ge  \frac{1}{K\tau}\sum_{k=0}^{K-1}     \theta^{-1}(H(\ba(k))). 
\end{align} Note that since $\theta$ is strictly increasing an concave, $\theta^{-1}$ is strictly increasing and convex.  Thus, by Jensen's inequality 
\begin{align}\label{eq:lb_relaxed_once}
   \frac{1}{K\tau}\sum_{k=0}^{K-1}     \theta^{-1}\left(H\left(\ba(k)\right)\right)  \ge  \frac{1}{\tau}   \theta^{-1}\left(\frac{1}{K}\sum_{k=0}^{K-1} H\left( \ba(k)\right)\right). 
\end{align}  We also have the following chain of inequalities 
\begin{IEEEeqnarray}{rCl}
H(\ba(k)) &\ge&  H(\ba(k)|\ba(0:k-1))\label{eq:l1pf_cre}\\ &\ge& I(\ba(k);\bx(0:K-1)|\ba(0:k-1)) \label{eq:l1pf_defmi}
\end{IEEEeqnarray} where (\ref{eq:l1pf_cre}) follows since conditioning reduces entropy, and (\ref{eq:l1pf_defmi}) follows from the definition of mutual information and the fact that
$H(\ba(k)|\bx(0:{K-1}),\ba(0:k-1))\ge0$. 
Summing (\ref{eq:l1pf_defmi}) and applying the chain rule for mutual information gives
\begin{align}
    \sum_{k=0}^{K-1} H( \ba(k))) \ge I(\bx(0:K-1);\ba(0:K-1)).
\end{align}  As $\bx(0:K-1)\leftrightarrow \ba(0:K-1)  \leftrightarrow \hat{\bx}(0:K-1)$ is a Markov chain, the data processing inequality gives 
\begin{multline}
    I(\bx(0:K-1);\ba(0:K-1)) \ge\\ I(\bx(0:K-1);\hat{\bx}(0:K-1)). 
\end{multline}
Finally, since $\theta^{-1}$ is increasing
\begin{multline}\label{eq:thelastonewecombine}
    \frac{1}{\tau}   \theta^{-1}(\frac{1}{K}\sum_{k=0}^{K-1} H( \ba(k))) \ge\\  \frac{1}{\tau}   \theta^{-1}(\frac{1}{K}I(\bx(0:K-1);\hat{\bx}(0:K-1))) 
\end{multline}
Combining (\ref{eq:the_lb_we_relax}),  (\ref{eq:lb_relaxed_once}), and (\ref{eq:thelastonewecombine}) proves the lemma. 
\end{proof} Lemma \ref{lemm:mutual_information_lb} will allow us to lower-bound (\ref{eq:optim_first_simp}) via an optimization from conventional causal DT rate-distortion theory \cite{tanaka2016semidefinite}. We presently discuss on of the key results from \cite{tanaka2015semidefinite}.

Define the sequence of Borel measurable causal \textit{reconstruction} kernels  \begin{align}\label{eq:reconkern}
\mathbb{P}_{\mathrm{R}}[\hat{\bx}(k)||\bx(k)] = \{\mathbb{P}[\hat{\bx}(k)|\hat{\bx}(0:k-1),\bx(0:k)]\}_{k\in\mathbb{N}_{0}}.
\end{align} Note that while any choice causal encoder kernel 
$\mathbb{P}\left[ \ba(k)| \ba(0:k-1), \bx(0:k)\right]$ and 
 decoder/interpolator $\hat{\bx}(k) = \mathbb{E}[\bx(k)|\ba(0:k)]$ induces a reconstruction kernel of the form (\ref{eq:reconkern}), the set of \textit{all} reconstruction policies (\ref{eq:reconkern}) is more general (e.g., it is not restricted to policies where the ``message" from the encoder to the estimation center is a discrete codeword $\ba(k)$). For $D>0$, $\dmat{Q}\in\mathbb{R}^{m\times m}$, $\dmat{Q}\succ 0$ , define the DT causal rate-distortion function 
\begin{align}\label{eq:rdf_midef}
R(D,\dmat{Q},\tau) = \inf_{\mathbb{P}_{\mathrm{R}}}  & \limsup_{K\rightarrow \infty}\frac{I(\bx(0:K-1);\hat{\bx}(0:K-1))}{K} \nonumber\\
\st & \limsup_{K\rightarrow \infty} \frac{1}{K}\sum_{k=0}^{K-1} \mathbb{E}[\lVert\rvec{e}\rVert_{\dmat{Q}}^2] \leq D
,
\end{align}
 where the expectations and information measures are computed with respect to the joint measure induced by the DT system dynamics  (\ref{eq:discretizeddynamics_vec}) and the reconstruction policy. It turns out that $R(D,\dmat{Q},\tau)$ can be computed via semidefinite programming \cite[Section V]{tanaka2016semidefinite}. $R(D,\dmat{Q},\tau)$ is equivalent to the (convex) log-determinant optimization 
\begin{align}\label{eq:RDF_disc_tc_vec_explicit}
R(D,\dmat{Q},\tau) &=& \left\{ \begin{aligned}
& \inf_{\substack{\dmat{P},\dmat{\Pi} \in \mathbb{R}^{m\times m}\\ \dmat{P},\dmat{\Pi} \succ 0 }} -\frac{1}{2}\log_{2}\left(\frac{\det\dmat{\Pi}}{\det (\dmat{B}_{\tau}\dmat{B}_{\tau}^{\top})}\right)\\ &\text{ }\st  \Tr(\dmat{Q}\dmat{P})\le D \text{,  }\\&\text{ }\text{\phantom{s.t.} }\dmat{P}\preceq \dmat{A}_{\tau}\dmat{P}\dmat{A}_{\tau}^{\top}+ \dmat{B}_{\tau}\dmat{B}_{\tau}^{\top} \text{, } \\&\text{ }\text{ }\text{ }\begin{bmatrix} \dmat{P}-\dmat{\Pi} & \dmat{P}\dmat{A}_{\tau}^{\top}\\ \dmat{A}_{\tau}\dmat{P} & \dmat{A}_{\tau}\dmat{P} \dmat{A}_{\tau}^{\top}+\dmat{B}_{\tau}\dmat{B}_{\tau}^{\top}\end{bmatrix}\succeq 0
\end{aligned}\right. 
\end{align}  The proof that (\ref{eq:RDF_disc_tc_vec_explicit}) is equivalent to (\ref{eq:rdf_midef}) follows from proving that the optimal policy reconstruction policy
(\ref{eq:reconkern}) can be realized via a two-stage architecture consisting of a linear-Gaussian sensor and a Kalman filter \cite{tanaka2016semidefinite}. It can be shown that if $\dmat{P}^{*}$ is the minimizing $\dmat{P}$ from (\ref{eq:RDF_disc_tc_vec_explicit}) then for $\dmat{P}^{*}_{+}= \dmat{A}_{\tau}\dmat{P}^{*}\dmat{A}_{\tau}^{\top}+ \dmat{B}_{\tau}\dmat{B}_{\tau}^{\top}$ we have 
\begin{align}\label{eq:rdf_explicit_at_minimum}
    R(D,\dmat{Q},\tau) = \frac{1}{2}\log_{2}\left(\frac{\det{\dmat{P}^{*}_{+}}}{\det{\dmat{P}^{*}}}\right).
\end{align}

The main result of this section is the following theorem. 
\begin{theorem}\label{thm:firstbound}
Define the strictly increasing, concave function $\theta(x):\mathbb{R}^{+}\rightarrow\mathbb{R}^{+}$ via $\theta(x) = x+(1+x)\log_2(1+x)-x\log_2(x)$, and denote its inverse as $\theta^{-1}:\mathbb{R}^{+}\rightarrow\mathbb{R}^{+}$. Let $D_{c}>0$. We have 
\begin{IEEEeqnarray}{rCl}\label{eq:rdf_lb}
      \mathcal{L}^*_{\mathrm{CT}}(D_{c},\tau) &\ge& 
       \mathcal{L}^*_{\mathrm{DT}}(D_{d,\tau},\overline{\dmat{Q}}_{\dmat{A},\tau},\tau)
      \\&\ge & \frac{1}{\tau}\theta^{-1}\left( R(D_{d,\tau},\overline{\dmat{Q}}_{\dmat{A},\tau},\tau)\right).
\end{IEEEeqnarray}
\end{theorem} 
\begin{proof}
    Since by Proposition \ref{prop:cd2dtprop} we have $  \mathcal{L}^*_{\mathrm{DT}}(D_{d,\tau},\overline{\dmat{Q}}_{\dmat{A},\tau},\tau) \le \mathcal{L}^*_{\mathrm{CT}}(D_{c},\tau)$, we will prove the inequality \begin{align}
        \frac{1}{\tau}\theta^{-1}\left( R(D_{d,\tau},\overline{\dmat{Q}}_{\dmat{A},\tau},\tau)\right) \le \mathcal{L}^*_{\mathrm{DT}}(D_{d,\tau},\overline{\dmat{Q}}_{\dmat{A},\tau},\tau). 
    \end{align} Via the definition of $\mathcal{L}^*_{\mathrm{DT}}(D,\dmat{Q},\tau)$ in (\ref{eq:optim_first_simp}) and Lemma \ref{lemm:mutual_information_lb}'s (\ref{eq:converse}) we have 
\begin{multline}
    \mathcal{L}^*_{\mathrm{DT}}(D,\dmat{Q},\tau) \ge\\  \left\{\begin{aligned} &\inf_{\mathbb{P}_{\mathrm{E}}} \limsup_{K\rightarrow \infty} \frac{1}{\tau}  \theta^{-1}\left(\frac{1}{K}I(\bx(0:K-1);\hat{\bx}(0:K-1))\right)
    \\ &\st \limsup_{K\rightarrow \infty} \frac{1}{K}\sum_{k=0}^{K-1} \mathbb{E}[\lVert \rvec{e}(k)\rVert_\dmat{Q}^2] \le D. \end{aligned} \right.
\end{multline} Since $\theta^{-1}$ is increasing and continuous, we have
\begin{multline}\label{eq:optim_first_simpprime}
    \mathcal{L}^*_{\mathrm{DT}}(D,\dmat{Q},\tau) \ge\\ \frac{1}{\tau}\theta^{-1}\left( \begin{aligned} &\inf_{\mathbb{P}_{\mathrm{R}}} \limsup_{K\rightarrow \infty} \frac{1}{K}I(\bx(0:K-1);\hat{\bx}(0:K-1))
    \\ &\st \limsup_{K\rightarrow \infty} \frac{1}{K}\sum_{k=0}^{K-1} \mathbb{E}[\lVert \rvec{e}(k)\rVert_\dmat{Q}^2] \le D. \end{aligned}\right)   
\end{multline}
Note that the right hand optimization (\ref{eq:optim_first_simpprime}) depends only on the $\rvec{x}(k)$ and $\rvec{\hat{x}}(k)$, and that with respect to (\ref{eq:optim_first_simpprime}), the policy space $\mathbb{P}_{\mathrm{R}}$ is richer than $\mathbb{P}_{\mathrm{E}}$ as it includes all causal kernels from $\rvec{x}(k)$ to $\rvec{\hat{x}}(k)$ that can be induced by $\mathbb{P}_{\mathrm{E}}$. Comparing the right-hand optimization in (\ref{eq:optim_first_simpprime}) to the definition of $R(D,\dmat{Q},\tau)$ in (\ref{eq:rdf_midef}), it is clear that  $\mathcal{L}^*_{\mathrm{DT}}(D,\dmat{Q},\tau)\ge \frac{1}{\tau}\theta^{-1}(R(D,\dmat{Q},\tau))$. 
\end{proof} In the next subsections, we will establish Theorem \ref{thm:keybounds_abbrev} via relaxing Theorem \ref{thm:firstbound}. The main utility of this is that, given Lemma \ref{lemm:ct2sdp}, we only need to solve a single semidefinite program (\ref{eq:the_ct_ti_sdp}) to derive a lower on (\ref{eq:initialOptimization}) for all values of the sampling period $\tau$. This is in stark contrast to the bound in Theorem \ref{thm:firstbound}, which requires solving a separate semidefinite program, depending on the discretized state-space model, for each $\tau$. The proof of (\ref{eq:uselemmacont}) in Theorem \ref{thm:keybounds_abbrev} will follow upon demonstrating that (\ref{eq:rate_info_c}) lower bounds the time-normalized DT rate distortion function $\frac{1}{\tau}R(D_{d,\tau},\overline{\dmat{Q}}_{\dmat{A},\tau},\tau)$ (cf. (\ref{eq:RDF_disc_tc_vec_explicit})) for any $\tau$. This makes intuitive sense, given that (\ref{eq:RDF_disc_tc_vec_explicit}) allows one to optimize over a richer space of CT linear policies as opposed to (\ref{eq:RDF_disc_tc_vec_explicit})'s DT policies. In the next section, we will explore the structural properties of (\ref{eq:rate_info_c}) before we establish a connection between the CT optimization (\ref{eq:rate_info_c}) and the DT optimization (\ref{eq:rdf_midef}). 

\subsection{Convexifying (\ref{eq:rate_info_c}) and establishing Lemma \ref{lemm:ct2sdp} }\label{subsec:convexproofs}
While the optimization (\ref{eq:rate_info_c}) has the advantage of not depending on $\tau$, it is not apparent that it can be computed. We now establish Lemma \ref{lemm:ct2sdp}, which states that
(\ref{eq:rate_info_c}) can be computed via a semidefinite program.

Consider the Riccati differential equation (\ref{eq:riccati}) and the expressions for the CT mutual information (\ref{eq:mi}). Using (\ref{eq:riccati}) and applying the cyclic property of the trace several times, irrespective of the choice of $\dmat{C}_{t}\in \mathcal{F}_{\mathbb{R}\rightarrow \mathbb{R}^{m\times m}}$,  we have 
\begin{align}\label{eq:mipcvx}
  \Tr(\dmat{C}_{t} \dmat{X}_{t} \dmat{C}_{t}^\top) =  2  \Tr(\dmat{A}) -  \Tr(\dmat{X}_{t}^{-1}\dot{\dmat{X}}_{t}) +  \Tr(\dmat{B}^{\top}\dmat{X}_{t}^{-1}\dmat{B}).
\end{align} Let $\dmat{Y}_{t} = \dmat{B}^{\top}\dmat{X}_{t}^{-1}\dmat{B}$ and  $a = 2  \Tr(A)$, and recall  that $\frac{d}{dt}\ln(\det(\dmat{X}_{t})) =   \Tr(\dmat{X}^{-1}_{t}\dot{\dmat{X}_{t}})$ \cite[(43)]{cookbook}. Then, integrating (\ref{eq:mipcvx}) from $t=0$ to $T$ gives 
\begin{align}\label{eq:duncingbooth}
     I(\bx_{[0,T]};\hat{\bx}_{[0,T]}) = \frac{\left(Ta - \ln\left(\frac{\det(\dmat{X}_{T})}{\det(\dmat{X}_{0})}\right)+\int_{0}^{T}  \Tr\left(\dmat{Y}_{t}\right)dt \right)}{2\ln(2)}.
\end{align} Let $\tilde{\dmat{X}}_{t} = \dmat{A}_{t}\dmat{X}_{0}\dmat{A}^{\top}_{t}+ \dmat{B}_{t}\dmat{B}_{t}^{\top}$ be the solution to (\ref{eq:riccati}) when $\dmat{C}_{t} =0$ for all $t$. We have for all $T$,  
$\tilde{\dmat{X}}_{T} \succeq \dmat{X}_{T}$, where since $\dmat{A}$ is Hurwitz, $\dmat{R} =\lim_{T\rightarrow \infty} \tilde{\dmat{X}}_{T}$ is well-defined with $\dmat{R}\succ 0$. Thus for any choice of sensing policy
$ -\frac{1}{T}\ln({\det(\dmat{X}_{T})/\det(\dmat{X}_{0})}) 
\ge -\frac{1}{T}\ln({\det(\tilde{\dmat{X}}_{T})/\det(\dmat{X}_{0})})$,
 however $\lim_{T\rightarrow \infty }-\frac{1}{T}\ln\left(\frac{\det(\tilde{\dmat{X}}_{T})}{\det(\dmat{X}_{0})}\right) = 0$.
Thus
\begin{align}
     \underset{T\rightarrow\infty}{\lim \sup} \frac{1}{T}I(\bx_{[0,T]};\hat{\bx}_{[0,T]}) \ge  \frac{\left(a +  \underset{T\rightarrow\infty}{\lim \sup}\frac{1}{T}\int_{0}^{T}  \Tr\left(\dmat{Y}_{t}\right)dt \right)}{2\ln(2)}.\label{eq:misimpl}
\end{align} This observation enables a convexification of (\ref{eq:rate_info_c}). We have the following intermediate result.
\begin{lemma}\label{lemm:convexification1}
    We have that $\underline{\mathcal{I}}^c(D) \le \mathcal{I}^c(D) $ where
    \begin{subequations}
\label{eq:pre_sdp}
\begin{align}
\underline{\mathcal{I}}^c(D) = \inf_{X_{t} \succ 0, Y_{t}\succeq 0} \quad &  \frac{a+ \underset{T\rightarrow\infty}{\lim \sup}\frac{1}{T}\int_{0}^{T}  \Tr\left(\dmat{Y}_{t}\right)dt}{2\ln(2)}\label{eq:traceobjective} \\
    \st& \dmat{X}_t\text{ differentiable a.e. } \\
    &\quad \dot{\dmat{X}}_t  \overset{\mathrm{a.e.}}{\preceq} \dmat{A}\dmat{X}_{t}+\dmat{X}_{t}\dmat{A}^\top+\dmat{B}\dmat{B}^\top,\nonumber \\
    &\quad \dmat{X}_{0}=\Sigma_{0} \label{eq:deineq} \\
    & \begin{bmatrix}
    \dmat{Y}_{t} & \dmat{B}^{\top} \\ \dmat{B} & \dmat{X}_{t}
    \end{bmatrix}\succeq 0 \label{eq:shurinequality}\\
    & \limsup_{T\rightarrow\infty }\frac{1}{T}\int_0^T  \Tr(\dmat{X}_{t}) dt \le D.\label{eq:matrixse}
\end{align}
\end{subequations}
\end{lemma}
\begin{proof} Any trajectory $\dmat{C}_{t}\in \mathcal{F}_{\mathbb{R}\rightarrow \mathbb{R}^{m\times m}} $ in (\ref{eq:rate_info_c}) uniquely specifies $\dmat{X}_{t}$ via the initial value problem (\ref{eq:riccati}). By the definition of $\mathcal{F}_{\mathbb{R}\rightarrow \mathbb{R}^{m\times m}}$, $\dmat{X}_{t}\succ 0$. Define $ \dmat{Y}_{t}=\dmat{B}^{\top}\dmat{X}_{t}^{-1}\dmat{B}$. We will show that $(\dmat{X}_{t},\dmat{Y}_{t})$ is a feasible point in (\ref{eq:pre_sdp}) and attains the same value of the objective in  (\ref{eq:rate_info_c}). By (\ref{eq:riccati}), the induced $\dmat{X}_{t}$ satisfies (\ref{eq:deineq}) and the inequality constraint that proceeds it in (\ref{eq:pre_sdp}). By Shur complements, (\ref{eq:shurinequality}) is equivalent to 
\begin{align}\label{eq:shurequation}
    \dmat{Y}_{t} \succeq \dmat{B}^{\top}\dmat{X}_{t}^{-1}\dmat{B},
\end{align} which is satisfied under the prescribed choice of $\dmat{Y}_{t}$. The distortion constraint (\ref{eq:matrixse}) is identical to  (\ref{eq:secost}); and a feasible choice of $\dmat{C}_{t}$ in (\ref{eq:rate_info_c}) ensures that the induced  $\dmat{X}_{t}$ satisfies (\ref{eq:matrixse}). Thus, under this choice of $(\dmat{X}_{t},\dmat{Y}_{t})$, the mutual information in (\ref{eq:rate_info_c}) is at least as large as the right hand side of (\ref{eq:misimpl}), which is the objective function in (\ref{eq:pre_sdp}). Thus since any choice $\dmat{C}_{t}\in\mathcal{F}_{\mathbb{R}\rightarrow \mathbb{R}^{m\times m}}$ in (\ref{eq:rate_info_c}) produces a feasible $(\dmat{X}_{t},\dmat{Y}_{t})$ in (\ref{eq:pre_sdp}) that attains a lower bound on the cost that the $\dmat{C}_{t}\in\mathcal{F}_{\mathbb{R}\rightarrow \mathbb{R}^{m\times m}}$ attains in (\ref{eq:rate_info_c}), we have that (\ref{eq:pre_sdp}) lower bounds (\ref{eq:rate_info_c}). 
\end{proof} Lemma \ref{lemm:convexification1} shows that  (\ref{eq:rate_info_c}) has a convex lower bound. The lower bound itself, however, remains an optimization over an infinite dimensional space, and furthermore the tightness of $\underline{\mathcal{I}}^c(D)$ with respect to ${\mathcal{I}}^c(D)$ remains unclear. In the proof of Lemma \ref{lemm:ct2sdp}, we demonstrate that, in fact, ${\mathcal{I}}^c(D)=\underline{\mathcal{I}}^c(D)$, and furthermore one can compute ${\mathcal{I}}^c(D)$ via a semidefinite program. This will follow from a demonstration that one can achieve the minimum in (\ref{eq:rate_info_c}) with a time-invariant sensor gain. This leads to Lemma \ref{lemm:ct2sdp}, which we presently prove. The result follows analogously to a similar, DT result in \cite{tanaka2015semidefinite}. 
\begin{proof}[Proof of Lemma \ref{lemm:ct2sdp}]
    From Lemma \ref{lemm:convexification1}, we have that $\underline{\mathcal{I}}^c(D) \le {\mathcal{I}}^c(D)$  where $\underline{\mathcal{I}}^c(D) $ is as given in (\ref{eq:pre_sdp}). Since (\ref{eq:shurinequality}) is equivalent to (\ref{eq:shurequation}) given $\dmat{X}_{t}\succ 0$ and the objective (\ref{eq:traceobjective}) is monotonic in $\dmat{Y}_{t}$, we have that (\ref{eq:pre_sdp}) is equivalent to 
       \begin{subequations}
\label{eq:pre_sdp2}
\begin{align}
\underline{\mathcal{I}}^c(D)  = \inf_{X_{t} \succ 0} \quad &  \frac{a+ \underset{T\rightarrow\infty}{\lim \sup}\frac{1}{T}\int_{0}^{T}  \Tr\left(\dmat{B}^{\top}\dmat{X}^{-1}_{t}\dmat{B}\right)dt}{2\ln(2)}\label{eq:traceobjective2} \\
    \st& \dmat{X}_t\text{ differentiable a.e. } \\&\quad \dot{\dmat{X}}_t  \overset{\mathrm{a.e.}}{\preceq} \dmat{A}\dmat{X}_{t}+\dmat{X}_{t}\dmat{A}^\top+\dmat{B}\dmat{B}^\top, \label{eq:depineq2}\\
    &\quad \dmat{X}_{0}=\Sigma_{0} \label{eq:deineq2} \\
    & \limsup_{T\rightarrow\infty }\frac{1}{T}\int_0^T   \Tr(\dmat{X}_{t}) dt \le D.\label{eq:matrixse2}
\end{align}
\end{subequations} Assume for now that $\dmat{X}^{*}_{t}$ is the feasible minimizer of (\ref{eq:pre_sdp2}). We will consider the case that the infimum is not achieved by a feasible point later. Define the partial time average via
\begin{align}
    \overline{\dmat{X}}_{T}= \frac{1}{T}\int_{0}^{T}\dmat{X}^{*}_{t}dt,
\end{align} which depends on the horizon $T$. $ \overline{\dmat{X}}_{T}$ is a net (with respect to $T\in\mathbb{R}^{+}$). Integrating the equation preceding (\ref{eq:deineq2}) gives
\begin{multline}
    \frac{\dmat{X}^{*}_{T}-\dmat{X}^{*}_{0}}{T}  \preceq \dmat{A}\overline{\dmat{X}}_{T}+\overline{\dmat{X}}_{T}\dmat{A}^{\top}+\dmat{B}\dmat{B}^{\top} \Rightarrow \\ \frac{-1}{T}\dmat{\Sigma}^{*}_{0} \preceq \dmat{A}\overline{\dmat{X}}_{T}+\overline{\dmat{X}}_{T}\dmat{A}^{\top}+\dmat{B}\dmat{B}^{\top},
\end{multline} where the relaxation follows from the fact that $\dmat{X}^{*}_{0}=\dmat{\Sigma}^{*}_{0}$ and $\dmat{X}^{*}_{T}\succeq 0$. Thus for any $\epsilon>0$, there exists a $T'_{\epsilon}$ such that 
\begin{align}
    -\frac{\epsilon}{T}\dmat{I} \preceq \dmat{A}\overline{\dmat{X}}_{T}+\overline{\dmat{X}}_{T}\dmat{A}^{\top}+\dmat{B}\dmat{B}^{\top}\text{ }\forall\text{ }T\ge T'_{\epsilon}. 
\end{align} Likewise, by (\ref{eq:matrixse2}), for any $\epsilon>0$ there exists $T''_{\epsilon}$ such that 
\begin{IEEEeqnarray}{rCl}
         \Tr\left( \overline{\dmat{X}}_{T}\right) &=& \frac{1}{T}\int_0^T   \Tr(\dmat{X}_{t}) dt \\ &\le& D+\epsilon\text{ for all }T\ge T''_{\epsilon}.
\end{IEEEeqnarray}  Denote the compact set $\mathcal{C}_{\epsilon} = \left\{\vphantom{-\frac{\epsilon}{T}\dmat{I} }   \dmat{X}\succeq 0 :   \Tr\left(\dmat{X}\right)\le D+\epsilon, \right. \\ \left.-\frac{\epsilon}{T}\dmat{I} \preceq \dmat{A}\dmat{X}+\dmat{X}\dmat{A}^{\top}+\dmat{B}\dmat{B}^{\top} \right\}$ and define $T_{\epsilon} = \max(T'_{\epsilon}, T''_{\epsilon})$. For all $T\ge T_{\epsilon}$, $\overline{\dmat{X}}_{T}\in \mathcal{C}_{\epsilon} $. Fix some $\epsilon>0$. Since $\mathcal{C}_{\epsilon}$ is compact and $\overline{\dmat{X}}_{T}\in \mathcal{C}_{\epsilon}$ for all $T$ sufficiently large, $\overline{\dmat{X}}_{T}$ has a cluster point in $\mathcal{C}_{\epsilon}$, e.g. there exists $\overline{\dmat{X}}^{*}\in \mathcal{C}_{\epsilon}$, such if $U$ is any open neighborhood of $\overline{\dmat{X}}^{*}$, for any $T_{0}\in\mathbb{R}^{+}$ there exists a $T>T_{0}$ such that $\overline{\dmat{X}}_{T}\in U$. We claim that the cluster point necessarily has  $\overline{\dmat{X}}^{*} \in \underset{{N\in\mathbb{N}}}{\cap}\mathcal{C}_{\frac{1}{N}}$ since for any $N$, there exists $J$ such that $\overline{\dmat{X}}_{T}$ is contained in $\mathcal{C}_{\frac{1}{N}}$ for $T\ge J$. Thus, we claim that $\overline{\dmat{X}}^{*}$ satisfies 
\begin{IEEEeqnarray}{rCl}
0 &\preceq& \dmat{A}\overline{\dmat{X}}^{*}+\overline{\dmat{X}}^{*}\dmat{A}^{\top}+\dmat{B}\dmat{B}^{\top}\text{ and }  \Tr(\overline{\dmat{X}}^{*})\le D.
\end{IEEEeqnarray}
Note finally that for any $T$, by Jensen's inequality and the convexity of $  \Tr(\dmat{B}^{\top}\dmat{X}^{-1}\dmat{B})$ we have
\begin{align}
      \Tr(\dmat{B}^{\top}\overline{\dmat{X}}_{T}^{-1}\dmat{B})\le \frac{1}{T}\int_{0}^{T}  \Tr(\dmat{B}^{\top}\left({\dmat{X}}^{*}_{t}\right)^{-1}\dmat{B})dt
\end{align} and since $\overline{\dmat{X}}^{*}$ is a cluster point of $\overline{\dmat{X}}_{T}$ we have
\begin{IEEEeqnarray}{rCl}
      \Tr(\dmat{B}^{\top}(\overline{\dmat{X}}^{*})^{-1}\dmat{B})&\le & \limsup_{T\rightarrow \infty}   \Tr(\dmat{B}^{\top}(\overline{\dmat{X}}_{T})^{-1}\dmat{B})\\ &\le & \limsup_{T\rightarrow \infty} \frac{1}{T}\int_{0}^{T}  \Tr(\dmat{B}^{\top}\left({\dmat{X}}^{*}_{t}\right)^{-1}\dmat{B})dt\nonumber 
\end{IEEEeqnarray} Thus we claim that if (\ref{eq:pre_sdp2}) is minimized by a feasible point, (\ref{eq:pre_sdp2}) is lower bounded by the following finite dimensional optimization
       \begin{subequations}
\label{eq:pre_sdp3}
\begin{align}
\underline{\underline{\mathcal{I}}}^c(D) = \inf_{\dmat{X} \succeq 0} \quad &  \frac{a+   \Tr\left(\dmat{B}^{\top}\dmat{X}^{-1}\dmat{B}\right)}{2\ln(2)}\label{eq:traceobjective3} \\
    \st& 0  \preceq \dmat{A}\dmat{X}+\dmat{X}\dmat{A}^\top+\dmat{B}\dmat{B}^\top,\\
    &   \Tr(\dmat{X}) \le D.\label{eq:matrixse3}
\end{align}
\end{subequations} 
In the case that the infimum in (\ref{eq:pre_sdp2}) is not achieved by a feasible point, there must exist a sequence of feasible trajectories that achieve a value in the objective (\ref{eq:traceobjective2}) arbitrarily close to the infimum. One can demonstrate that (\ref{eq:pre_sdp3}) lower bounds the value of (\ref{eq:traceobjective2}) achieved by each element of this sequence. In other words, we have $\underline{\underline{\mathcal{I}}}^c(D)\le \underline{\mathcal{I}}^c(D) \le  {\mathcal{I}}^c(D)$. Note that in (\ref{eq:pre_sdp3}), $\dmat{X} \succ 0$ holds without loss of generality since the objective function (\ref{eq:traceobjective3}) tends to infinity for near-singular $\dmat{X}$. These results have a useful interpretation.  Let $\tilde{\dmat{X}}_{t}$ be a feasible choice of $\dmat{X}_{t}$ in (\ref{eq:pre_sdp2}) such that the infinite-horizon time-average $\lim_{T\rightarrow \infty}\int_{0}^{T}\tilde{\dmat{X}}_{t}dt/T = \overline{\dmat{X}}$ is well defined. Replacing $\dmat{X}_{t}=\overline{\dmat{X}}$ for all $t$ will results in a reduced mutual information cost (\ref{eq:traceobjective2}) with respect to that incurred by $\dmat{X}_{t}=\tilde{\dmat{X}}_{t}$, will attain the same value of the distortion constraint (\ref{eq:matrixse2}), and will satisfy (\ref{eq:depineq2}). The constrained initial condition (\ref{eq:deineq2}), however, will not necessarily be met.  We can think of (\ref{eq:pre_sdp3}) as searching for the best ``time-invariant" choice of $\dmat{X}_{t}$ that satisfies the constraints of (\ref{eq:pre_sdp2}), neglecting that initial condition. By \Cesaro means, choosing $\dmat{X}_{t}$ to be asymptotically equal to $\overline{\dmat{X}}$ will incur the same rate and distortion performance as the time-invariant solution. In the next paragraph, we will demonstrate how to design a trajectory $\dmat{X}_{t}$ that satisfies both the initial condition (\ref{eq:deineq2}) and (\ref{eq:depineq2}) and also has $\lim_{t\rightarrow\infty}\dmat{X}_{t} = \overline{\dmat{X}}$. We do this via solving (\ref{eq:riccati}) under a particular, time-invariant choice of measurement matrix $\dmat{C}_{t}$. 

We now demonstrate that in fact, ${\mathcal{I}}^c(D)=\underline{\underline{\mathcal{I}}}^c(D)$. To see this, we will prove that ${\mathcal{I}}^c(D) \le \underline{\underline{\mathcal{I}}}^c(D) $ by returning to the CT viewpoint in (\ref{eq:rate_info_c}). Let $\dmat{X}^{*}$ be the minimizing $\dmat{X}$ in (\ref{eq:pre_sdp3}), and let $\dmat{C}$ be any matrix that satisfies
\begin{align}\label{eq:tisensordef}
    \dmat{C}^{\top}\dmat{C} = (\dmat{X}^{*})^{-1}\left(\dmat{A}\dmat{X}^{*}+\dmat{X}^{*}\dmat{A}^\top+\dmat{B}\dmat{B}^\top\right)(\dmat{X}^{*})^{-1}.
\end{align} Note that the right hand side of (\ref{eq:tisensordef}) is positive semidefinite. Consider the policy $\dmat{C}_{t} = \dmat{C}$ for all $t$ in (\ref{eq:rate_info_c}). This choice of $\dmat{C}_{t}$ is locally integrable, and thus  $\dmat{C}_{t}\in  \mathcal{F}_{\mathbb{R}\rightarrow \mathbb{R}^{m\times m}}$. We will now demonstrate that this choice of $\dmat{C}_{t}$ is also feasible with respect to the distortion constraint (\ref{eq:secost}), and that it achieves a cost in (\ref{eq:micost}) that is equal to $\underline{\underline{\mathcal{I}}}^c(D)$. This will follow from a proof that under this choice of $\dmat{C}_{t}$, $\lim_{t\rightarrow \infty } \dmat{X}_{t} = \dmat{X}^{*}$.

Under this choice of $\dmat{C}_{t}$, the error covariance $\dmat{X}_{t}=\mathbb{E}\|\bx_t-\hat{\bx}_t\|^2$ satisfies the initial value problem 
\begin{multline}
\dot{\dmat{X}}_t=\dmat{A}\dmat{X}_{t}+\dmat{X}_{t}\dmat{A}^\top+\dmat{B}\dmat{B}^\top-\\ \dmat{X}_{t}\left( (\dmat{X}^{*})^{-1}\left(\dmat{A}\dmat{X}^{*}+\dmat{X}^{*}\dmat{A}^\top+\dmat{B}\dmat{B}^\top\right)(\dmat{X}^{*})^{-1} \right) \dmat{X}_{t}.
\end{multline} with $\dmat{X}_0=\dmat{\Sigma}_{0}$ and $\dmat{\Sigma}_{0}\succ 0$. Note that $\dot{\dmat{X}}_{t}=0$ when $\dmat{X}_{t} = \dmat{X}^{*}$. Using \cite[Section 3]{potterReport}, it can be shown that $\dmat{X}_{t} = \dmat{X}^{*}$ is the unique positive semidefinite solution of $\dmat{A}\dmat{X}_{t}+\dmat{X}_{t}\dmat{A}^\top+\dmat{B}\dmat{B}^\top-\\ \dmat{X}_{t}\left( (\dmat{X}^{*})^{-1}\left(\dmat{A}\dmat{X}^{*}+\dmat{X}^{*}\dmat{A}^\top+\dmat{B}\dmat{B}^\top\right)(\dmat{X}^{*})^{-1} \right) \dmat{X}_{t} = 0 $, and, in fact, given that $\Sigma_{0}\succ 0$, we have that $\lim_{t\rightarrow \infty } \dmat{X}_{t} = \dmat{X}^{*}$. To demonstrate this via the results in  \cite[Section 3]{potterReport}, we need only verify that \cite[pgs. 12, 25]{potterReport} (1)$ \dmat{B}\dmat{B}^\top\succeq 0$ and $\dmat{C}^{\top}\dmat{C}\succeq 0$, where $\dmat{C}^{\top}\dmat{C}$ is as in (\ref{eq:tisensordef}), (2) no (right) eigenvector of $\dmat{A}^{\top}$ with nonnegative real part is a null vector of $\dmat{B}\dmat{B}^{\top}$, and (3) that no (right) eigenvector of $\dmat{A}$ with nonnegative real part is a null vector of $\dmat{C}^{\top}\dmat{C}$. The first of these requirements is immediate via the optimization problem (\ref{eq:pre_sdp3}). The second follows since by assumption, $\dmat{B}\dmat{B}^{\top} \succ 0$. To prove the last requirement, assume that for $\dvec{v}\in\mathbb{R}^{m}$ we have $\dmat{A}\dvec{v} = \lambda\dvec{v}$ with $\mathcal{R}(\lambda)\ge 0$. We prove that $\dvec{v}$ is not a null vector of $\dmat{C}^{\top}\dmat{C}$ by contradiction. From (\ref{eq:tisensordef}), if $\dvec{v}^{\top}\dmat{C}^{\top}\dmat{C}\dvec{v} = 0$ then 
\begin{multline}\label{eq:tiproofbycontra}
\dvec{v}^{\top}(\dmat{X}^{*})^{-1}\dmat{B}\dmat{B}^\top(\dmat{X}^{*})^{-1}\dvec{v} =\\ -\dvec{v}^{\top}\left((\dmat{X}^{*})^{-1}\dmat{A}+\dmat{A}^\top(\dmat{X}^{*})^{-1}\right)\dvec{v}. 
\end{multline} Note that the left hand side of (\ref{eq:tiproofbycontra}) is strictly positive, since both $\dmat{B}\dmat{B}^\top \succ 0$ and $\dmat{X}^{*}\succ 0$. However, by assumption, the right hand side is equal to $-\dvec{v}^{\top}\left((\dmat{X}^{*})^{-1}\dmat{A}+\dmat{A}^\top(\dmat{X}^{*})^{-1}\right)\dvec{v} = -2\mathcal{R}(\lambda)\dvec{v}^{\top}(\dmat{X}^{*})^{-1}\dvec{v}$
which is nonpositive, since by assumption $\mathcal{R}(\lambda)\ge 0$. Thus (\ref{eq:tiproofbycontra}) is a contradiction; no (right) eigenvector of $\dmat{A}$ with nonnegative real part is a null vector of $\dmat{C}^{\top}\dmat{C}$. Thus $\lim_{t\rightarrow \infty } \dmat{X}_{t} = \dmat{X}^{*}$ \cite[pg. 25]{potterReport}. Via the \Cesaro mean, we have
\begin{IEEEeqnarray}{rCl}
    \lim_{T\rightarrow \infty}\frac{1}{T}\int_{0}^{T}   \Tr(\dmat{X}_{t})dt &=&    \Tr(\dmat{X}^{*})dt \\ &\le& D\label{eq:feasibilityinti}
\end{IEEEeqnarray} where (\ref{eq:feasibilityinti}) follows from the fact that $\dmat{X}=\dmat{X}^{*}$ is feasible with respect to (\ref{eq:matrixse3}) in (\ref{eq:pre_sdp3}). Consider the expression for the mutual information on the right-hand-side of (\ref{eq:duncingbooth}) (where by definition $\dmat{Y}_{t} = \dmat{B}^{\top}\dmat{X}_{t}^{-1}\dmat{B}$). Since $\lim_{t\rightarrow \infty } \dmat{X}_{t} = \dmat{X}^{*}$, we have
\begin{multline}
    \underset{T\rightarrow\infty}{\limsup} \frac{\left(- \ln\left(\frac{\det(\dmat{X}_{T})}{\det(\dmat{X}_{0})}\right)+\int_{0}^{T}  \Tr\left(\dmat{B}^{\top}\dmat{X}_{t}^{-1}\dmat{B}\right)dt \right)}{2\ln(2)T} =\\ \frac{ \Tr\left(\dmat{B}^{\top}(\dmat{X}^{*})^{-1}\dmat{B}\right) }{2\ln(2)T},
\end{multline} which follows via the \Cesaro mean. Thus (given (\ref{eq:duncingbooth}))
\begin{IEEEeqnarray}{rCl}
    \underset{T\rightarrow\infty}{\lim} \frac{1}{T}I(\bx_{[0,T]};\hat{\bx}_{[0,T]}) =  \frac{a +  \Tr\left(\dmat{B}^{\top}\left(\dmat{X}^{*}\right)^{-1}\dmat{B}\right)}{2\ln(2)}\label{eq:usecmeanandoptimality}.
\end{IEEEeqnarray} Note that (\ref{eq:usecmeanandoptimality}) is exactly the value of $\underline{\underline{\mathcal{I}}}^c(D)$ (i.e. (\ref{eq:usecmeanandoptimality}) is equal to (\ref{eq:traceobjective3}) at the minimizer $\dmat{X}= \dmat{X}^{*}$). Thus,  choosing $\dmat{C}_{t} = \dmat{C}$ with $\dmat{C}$ as given in (\ref{eq:tisensordef}) is a feasible policy in (\ref{eq:rate_info_c}) that attains an objective value equal to $\underline{\underline{\mathcal{I}}}^c(D)$. Thus,  ${\mathcal{I}}^c(D)\le \underline{\underline{\mathcal{I}}}^c(D)$. Since by our comments above (\ref{eq:pre_sdp3}), $ \underline{\underline{\mathcal{I}}}^c(D)\le\underline{\mathcal{I}}^c(D)\le  {\mathcal{I}}^c(D)$, it must be that  ${\mathcal{I}}^c(D)=\underline{\mathcal{I}}^c(D)=\underline{\underline{\mathcal{I}}}^c(D)$.

It remains to write (\ref{eq:pre_sdp3}) as a standard form semidefinite optimization. In analogy to our work in the proof of Lemma \ref{lemm:convexification1} (cf. around (\ref{eq:shurequation})) we have that (\ref{eq:pre_sdp3}) is equivalent to (\ref{eq:the_ct_ti_sdp}). 
\end{proof}
While we can now compute (\ref{eq:rate_info_c}) via Lemma \ref{lemm:ct2sdp}, we have yet to establish an operational interpretation of  (\ref{eq:rate_info_c}). This is the subject of the next section, where we establish Theorem \ref{thm:keybounds_abbrev}. 
\subsection{Proof of Theorem \ref{thm:keybounds_abbrev}}\label{subsec:ctfinalproofs}
Theorem \ref{thm:keybounds_abbrev} is an almost immediate consequence of the following result, which explicitly establishes a relationship between the DT rate-distortion function (\ref{eq:rdf_midef}) (or, equivalently (\ref{eq:RDF_disc_tc_vec_explicit})) and the CT information-distortion function (\ref{eq:rate_info_c}).  
\begin{lemma}\label{lemm:continuousrdflb}
Assume a CT distortion constraint $D_{c}>0$ and any sampling interval $\tau>0$ sufficiently small such that  $D_{d,\tau} = D_{c}\tau-\overline{b_{\tau}}$ has $D_{d,\tau}>0$ (c.f. Lemma \ref{lemm:dimplication_vec}). We have
\begin{align}\label{eq:rdf_to_continuous}
\frac{1}{\tau}R(D_{d,\tau},\overline{\dmat{Q}}_{\dmat{A},\tau},\tau) \ge \mathcal{I}^c(D_{c}).
\end{align}
\label{lemma:newlb}
\end{lemma}\begin{proof}
The strategy is to find a CT policy whose information-distortion function is arbitrarily close in performance to the equivalent DT rate-distortion function for a given $\tau$. Conceptually, the CT policy we will design is $\tau$-periodic where the first segment of each period consists of no observation and the second segment of each period is selected to force the error covariance along a linear trajectory that achieves the desired periodicity. The length of the second segment is then made sufficiently small to approach the behavior of the DT policy. We design the periodic policy assuming an artificial initial condition, and then use Lemma \ref{lemm:ct2sdp} to argue that a (CT) time-invariant policy exists that achieves (at worst) the same rate/distortion performance for any initial condition $\dmat{X}_0\succ 0$. 

Given $\varepsilon>0$, consider the SDP (\ref{eq:RDF_disc_tc_vec_explicit}), but with the following more strict constraint: 
\begin{equation}\label{eq:strict_SDP_con}
    \dmat{P} \preceq \dmat{A}_\tau \dmat{P} \dmat{A}_\tau^\top+\dmat{B}_\tau \dmat{B}_\tau^\top - \varepsilon I.
\end{equation} 
The optimal solution of this perturbed SDP is denoted $R(D,\dmat{Q},\tau,\varepsilon)$. Let $\dmat{P}^* \succ 0$ be the minimizing $\dmat{P}$ for the SDP corresponding to $R(D_{d,\tau},\overline{\dmat{Q}}_{\dmat{A},\tau},\tau,\varepsilon)$. Assume for now that $\dmat{X}_0 = \dmat{P}^*$. In the absence of observation ($\dmat{C}_t = 0$), we have that the solution to (\ref{eq:riccati}) is given by $\dmat{X}_t = \Tilde{\dmat{X}}_t := \dmat{A}_t \dmat{P}^* \dmat{A}_t^\top + \dmat{B}_t\dmat{B}_t^\top$. By (\ref{eq:strict_SDP_con}), we have that $\dmat{P}^* \prec \Tilde{\dmat{X}}_\tau$. Using this fact and the continuity of $\Tilde{\dmat{X}}_t$, we have that there exists a $\Delta > 0$ small enough such that for this initial condition and unobserved policy, $\dmat{P}^* \prec \Tilde{\dmat{X}}_{\tau - \Delta}$.

If $\tau-\Delta$ is chosen as the end of the unobserved segment of the CT policy, then the desired time-linear trajectory for the second segment, $t\in [\tau-\Delta, \tau)$, is given by
\begin{equation}\label{eq:des_lin_traj}
    \dmat{F}_t := \Tilde{\dmat{X}}_{\tau-\Delta} + \dmat{M}(t-(\tau-\Delta)), 
\end{equation}
where $\dmat{M} := (\dmat{P}^* - \Tilde{\dmat{X}}_{\tau-\Delta})/\Delta \prec 0$. Note that for $t\in [\tau-\Delta, \tau)$, $0\prec \dmat{P}^*\prec \dmat{F}_t \prec \Tilde{\dmat{X}}_\tau$, so $\dmat{F}_t^{-1}$ exists and $\dmat{F}_t$ is bounded on this interval. Thus, the following matrix is well-defined for $t\in [\tau-\Delta, \tau)$:
\begin{equation}\label{eq:policy_squared}
    \dmat{G}_t := -\dmat{F}_t^{-1}(\dmat{M} - \dmat{A}\dmat{F}_t - \dmat{F}_t\dmat{A}^\top - \dmat{B}\dmat{B}^\top)\dmat{F}_t^{-1}.
\end{equation}
$\dmat{F}_t$ being bounded on this interval implies $\Vert \dmat{A}\dmat{F}_t + \dmat{F}_t\dmat{A}^\top + \dmat{B}\dmat{B}^\top\Vert$ is bounded. Furthermore, $\lim_{\Delta \rightarrow 0} \Tilde{\dmat{X}}_{\tau-\Delta} = \Tilde{\dmat{X}}_\tau$, so 
$\Delta$ can be chosen such that $\Vert \dmat{M} \Vert$ is arbitrarily large. Combining this with the fact that $\dmat{M} \prec 0$ for small $\Delta$, we can choose $\Delta$ small enough such that $\dmat{M} - \dmat{A}\dmat{F}_t - \dmat{F}_t\dmat{A}^\top - \dmat{B}\dmat{B}^\top \prec 0$, resulting in $\dmat{G}_t \succ 0$ for $t\in [\tau-\Delta, \tau)$. 

$\dmat{G}_t$ is real and symmetric over this interval, so we may apply the spectral theorem and the fact that $\dmat{G}_t \succ 0$ to conclude that there exists a matrix $\dmat{H}_t \in \mathbb{R}^{n\times n}$ such that $\dmat{G}_t = \dmat{H}_t^\top\dmat{H}_t$. The preceding argument shows that for any small $\varepsilon > 0$ there is a sufficiently small $\Delta > 0$ such that the following policy is well-defined and real over for $t \in [0,\tau)$:
\begin{equation}\label{eq:linearizing_policy}
    \dmat{C}_t = \begin{cases}
        0, & t \in [0,\tau-\Delta) \\
        \dmat{H}_t, & t \in [\tau-\Delta, \tau).
    \end{cases}
\end{equation}
This policy produces the following error covariance solution:
\begin{equation}\label{eq:err_cov_sol}
    \dmat{X}_t = \begin{cases}
        \Tilde{\dmat{X}}_t, & t \in [0,\tau-\Delta) \\
        \dmat{F}_t, & t \in [\tau-\Delta, \tau).
    \end{cases}
\end{equation} 

\begin{figure}[h]
    \centering \includegraphics[width=0.48\textwidth]{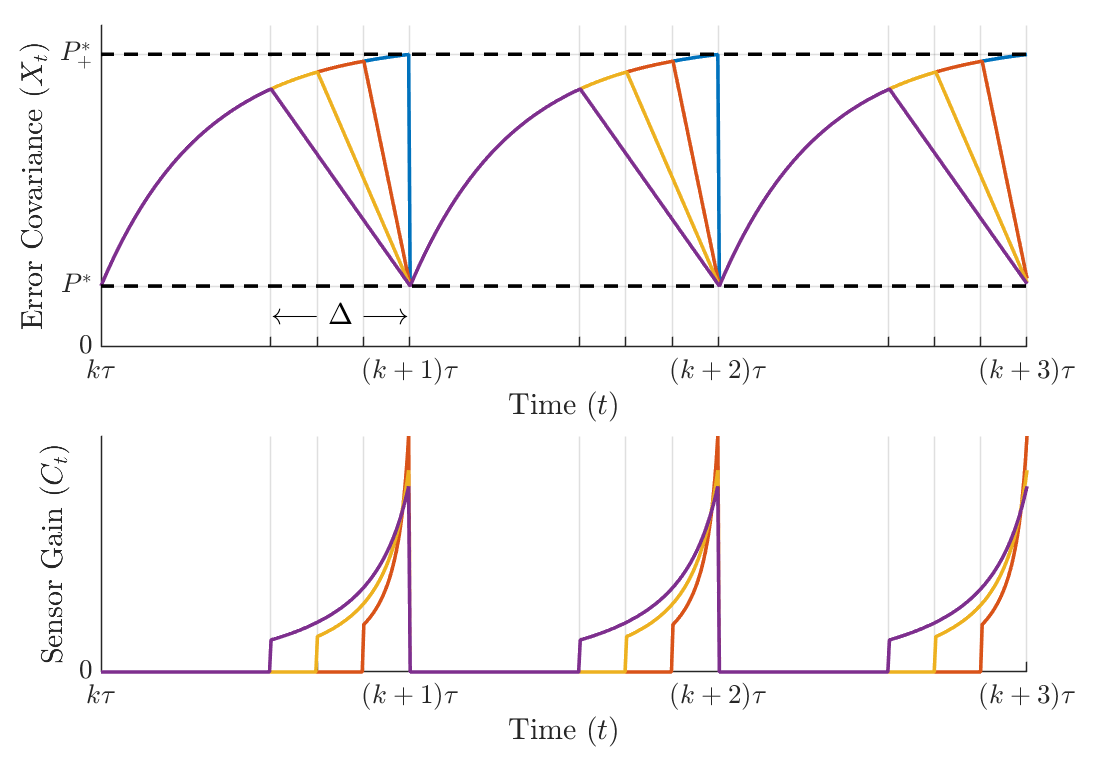}
    \caption{Diagram of the proposed CT observation policy in (\ref{eq:linearizing_policy}) and the corresponding covariance in (\ref{eq:err_cov_sol}) for some $D_{d,\tau}$ and $\varepsilon$ in the 1-dimensional case. The plots depict the change in the policy as $\Delta$ becomes smaller (purple, yellow, red, blue, respectively).}
    \label{fig:ctrdlb_obs_scheme}
\end{figure}
By (\ref{eq:mseX}), the mean squared error is given by
\begin{IEEEeqnarray}{rCl} 
    \IEEEeqnarraymulticol{3}{l}{\limsup_{T\rightarrow \infty} \frac{1}{T} \int_0^T   \Tr(\dmat{X}_t) dt} \nonumber\\ 
    &=&  \frac{1}{\tau} \left[ \int_0^{\tau-\Delta}   \Tr(\Tilde{\dmat{X}}_t) dt + \int_{\tau-\Delta}^\tau   \Tr(\dmat{F}_t) dt\right]  \\
    &=& \frac{1}{\tau} \left[ \int_0^{\tau-\Delta}   \Tr(\Tilde{\dmat{X}}_t) dt + \frac{\Delta}{2}  \Tr(\dmat{P}^*+\Tilde{\dmat{X}}_{\tau-\Delta})\right]  \label{eq:prelim_mse}
\end{IEEEeqnarray}
In the limit of (\ref{eq:prelim_mse}) as $\Delta \rightarrow 0$, the MSE is given by
\begin{align}
    \frac{1}{\tau}\int_0^\tau   \Tr (\Tilde{\dmat{X}}_t) dt &= \frac{1}{\tau}\left(  \Tr(\overline{\dmat{Q}}_{\dmat{A},\tau}\dmat{P}^{*})+ \overline{b}_{\tau} \right)\\
    &\leq \frac{1}{\tau}\left( D_{d,\tau}+\overline{b}_{\tau}\right) = D_c,
\end{align}
where the inequality follows from the fact that $\dmat{P^*}$ is the minimizing $\dmat{P}$ for the SDP corresponding to $R(D_{d,\tau},\overline{\dmat{Q}}_{\dmat{A},\tau},\tau,\varepsilon)$.

We now analyze the mutual information of this policy. Via (\ref{eq:mi}), (\ref{eq:mipcvx}), and the logic leading to (\ref{eq:duncingbooth}) we have that the mutual information for this policy is given by the following expression scaled by a factor of $(2\ln 2)^{-1}$ 
\begin{multline} 
    \limsup_{T\rightarrow \infty} \frac{1}{T} \int_0^T   \Tr(\dmat{C}_t\dmat{X}_t\dmat{C}_t^\top) dt  
    =\\ \frac{1}{\tau}\left[ \ln\left( \frac{\det(\Tilde{\dmat{X}}_{\tau-\Delta})}{\det(\dmat{P}^*)} \right) + \int_{\tau-\Delta}^\tau   \Tr(2\dmat{A}+\dmat{F}_t^{-1}\dmat{B}\dmat{B}^{\top}) dt\right] \label{eq:prelim_mi}
\end{multline}
For any $\Delta$, we have $\dmat{P}^*\prec \dmat{F}_t \prec \Tilde{\dmat{X}}_\tau$ for all $t \in [\tau-\Delta,\tau)$, so $  \Tr(2\dmat{A}+\dmat{F}_t^{-1}\dmat{B}\dmat{B}^{\top})$ is bounded by some constant for all $\Delta$, so the integral in (\ref{eq:prelim_mi}) goes to 0 as $\Delta \rightarrow 0$. Thus, in this limit, we find that the mutual information is
\begin{equation}
    \frac{1}{\tau}\frac{1}{2}\log_2\left(\frac{\det(\Tilde{\dmat{X}}_{\tau})}{\det(\dmat{P}^*)}\right) = \frac{1}{\tau}R(D_{d,\tau},\overline{\dmat{Q}}_{\dmat{A},\tau},\tau,\varepsilon),
\end{equation} where we used (\ref{eq:rdf_explicit_at_minimum}) and the fact that $\Tilde{\dmat{X}}_{\tau} = \dmat{P}^{*}_{+}$ as defined there. The unperturbed SDP (\ref{eq:RDF_disc_tc_vec_explicit}) is convex, so the perturbed value function, $R(D_{d,\tau},\overline{\dmat{Q}}_{\dmat{A},\tau},\tau,\varepsilon)$, is a convex function of $\varepsilon$ \cite[\S 5.6.1]{boyd2004convex} and is therefore continuous with respect to $\varepsilon$. Thus, the function can be made arbitrarily close to $R(D_{d,\tau},\overline{\dmat{Q}}_{\dmat{A},\tau},\tau) = R(D_{d,\tau},\overline{\dmat{Q}}_{\dmat{A},\tau},\tau,0)$ by choosing sufficiently small $\varepsilon$.

We will now eliminate the requirement that $\dmat{X}_0= \dmat{P}^*$. Let $\overline{\dmat{X}} = \frac{1}{\tau}\int_{0}^{\tau}\dmat{X}_{s}ds = \lim_{T\rightarrow \infty}\frac{1}{T}\int_{0}^{T}\dmat{X}_{s}ds$, where $\dmat{X}_{s}$ is the periodic solution in (\ref{eq:err_cov_sol}) we derived assuming $\dmat{X}_0= \dmat{P}^*$. Since $\dmat{P}^*\succ 0$,  $\overline{\dmat{X}}\succ 0$. Let $\dmat{C}^{\top}\dmat{C} = (\overline{\dmat{X}})^{-1}\left(\dmat{A}\overline{\dmat{X}}+\overline{\dmat{X}}\dmat{A}^\top+\dmat{B}\dmat{B}^\top\right)(\overline{\dmat{X}})^{-1}$. Following from the proof of Lemma \ref{lemm:ct2sdp}, it can be shown that choosing $\dmat{C}_{t} =\dmat{C}$ for all $t$ results in a policy that achieves a lower value of the mutual information objective (\ref{eq:micost}) than that of the periodic policy while attaining the same distortion cost (\ref{eq:secost}). This holds for any initial condition $\dmat{X}_{0}\succ 0$;  under this time-invariant policy, $\dmat{X}_{t}$ tends to $\overline{\dmat{X}}$ exponentially fast.
\end{proof}
Given Theorem \ref{thm:firstbound} and Lemma \ref{lemm:continuousrdflb}, Theorem \ref{thm:keybounds_abbrev} is immediate. In the next section, we perform a numerical study to gauge the tightness of the bounds in Theorem \ref{thm:keybounds_abbrev}.

\section{Numerical Experiments}\label{sec:numerical}
In this section, we simulate various encoding policies to explore the tightness of the bound in (\ref{eq:uselemmacont}).

\subsection{Scalar source with A-B scheme} We begin first by comparing the performance achievable by a variation of the {\AA}str\"om-Bernhardsson ({\AA}-B) scheme \cite{astrom2002comparison} to the lower bound in (\ref{eq:uselemmacont}) for scalar systems. Consider a system governed by (\ref{eq:source_c}) with $n = 1$ and $\rvec{x}_0 = 0$. Under the original ({\AA}-B) scheme, at any time $t$ such that $|\rvec{x}_t - \hat{\rvec{x}}_t| = d$ where $d>0$ is a fixed threshold, the encoder sends a bit that indicates the sign of the error and the estimate is updated accordingly. Unlike the ({\AA}-B) scheme which encodes messages continuously, we consider only encoders restricted to sending packets at discrete times. Therefore, we utilize the following modification. Let $\hat{\rvec{x}}(0) = 0$. Given $\hat{\rvec{x}}(k)$ and a fixed threshold $d>0$, update the estimate at time step $k+1$ as follows:
\begin{equation}
    \label{eq:modAB}
    \hat{\rvec{x}}(k+1)=
    \begin{cases}
    \dmat{A}_{\tau}(\hat{\rvec{x}}(k)+d) & \text{ if } \rvec{x}(k)-\hat{\rvec{x}}(k)\geq d \\
    \dmat{A}_{\tau}\hat{\rvec{x}}(k) & \text{ if } |\rvec{x}(k)-\hat{\rvec{x}}(k)| < d \\
    \dmat{A}_{\tau}(\hat{\rvec{x}}(k)-d) & \text{ if } \rvec{x}(k)-\hat{\rvec{x}}(k)\leq -d.
    \end{cases}
\end{equation}

For $n = 1$, the CT information-distortion function is given by:
\begin{equation}
    \label{eq:1d_ctid}
    \mathcal{I}^c(D) = (\log_2 e) \max \left\{ 0, \dmat{A} + \frac{\dmat{B}^2}{2D}\right\}.
\end{equation}
This provides an efficient method to compute the lower bound (\ref{eq:uselemmacont}) for the case of a scalar source. Figure~\ref{fig:ab_compare} compares this lower bound, $\theta^{-1}(\tau \mathcal{I}^c(D))/\tau$, to the simulated performance of the policy in (\ref{eq:modAB}) subject to the dynamics in (\ref{eq:discretizeddynamics_vec}).
\begin{figure}[h]
    \centering \includegraphics[width=0.48\textwidth]{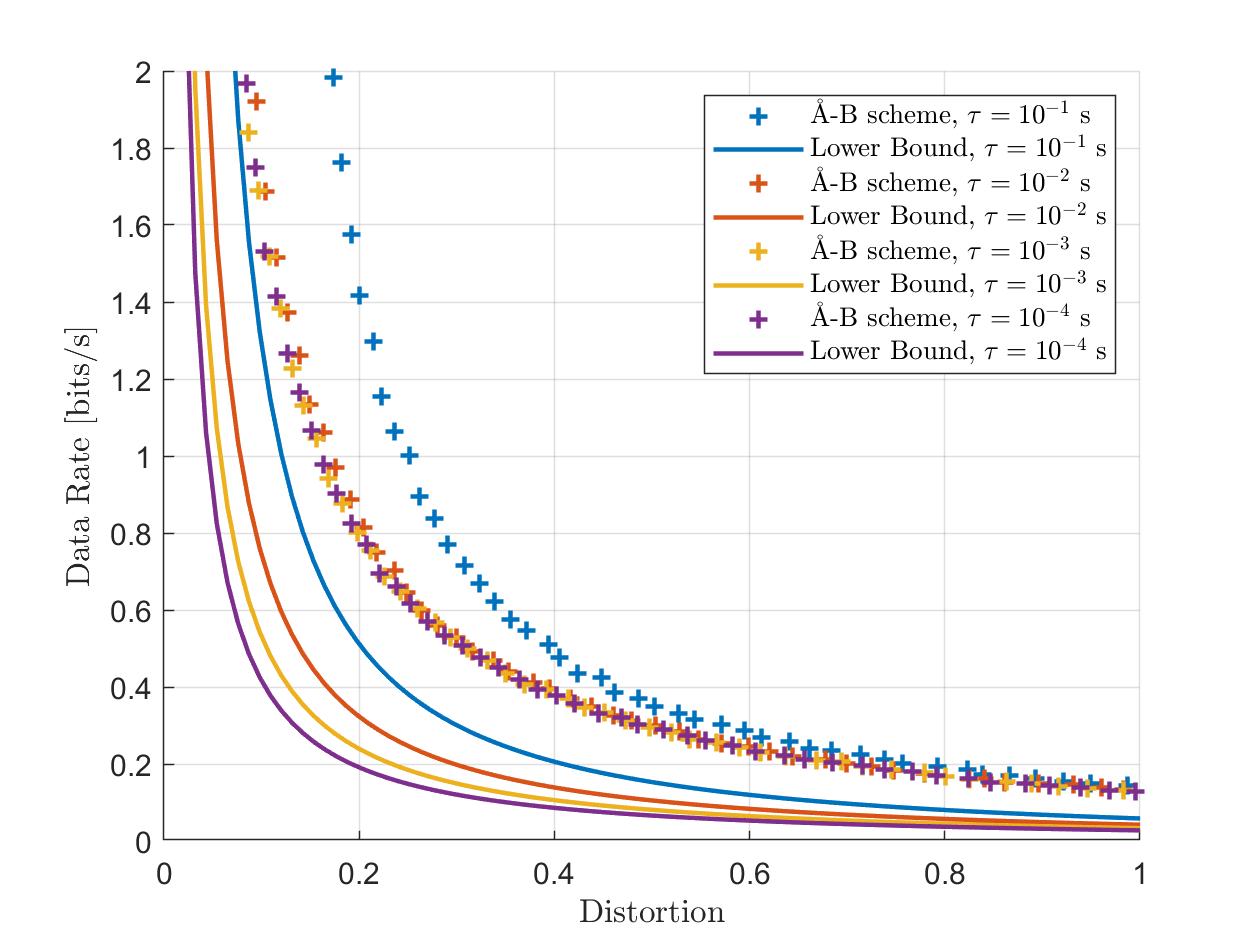}
    \caption{Comparison of the modified \AA-B scheme performance to the lower bound $\theta^{-1}(\tau \mathcal{I}^c(D))/\tau$. For this simulation, the system parameters were chosen to be $\dmat{A} = -0.1$ and $\dmat{B} = 1$. The threshold $d$ was varied between 0.5 and 2.5. Each simulation was run over a time horizon of $10000d^2$.}
    \label{fig:ab_compare}
\end{figure}

A key consequence of \cite{guo2021optimal_IT,guo2021optimal_TAC} is that the CT {\AA}-B approach is optimal for CT (i.e. $\tau\rightarrow 0$) minimum rate tracking with event-based sampling. We conjecture that in the limit of $\tau\searrow 0$ the rate-distortion performance achieved by the DT {\AA}-B scheme recovers that of the original CT version, although we are not aware of a formal proof of this claim. This suggests that for sufficiently small $\tau$ (e.g., $\tau=0.001$), the performance of the {\AA}-B scheme in Fig.~\ref{fig:ab_compare} (1) approximates the optimal continuous-time rate-distortion trade-off attainable by general real-time encoders/decoders and (2) provides relatively tight upper bounds on the rate-distortion tradeoff among systems with a minimum sampling period $\tau$. Note that in Fig. \ref{fig:ab_compare}, the empirical rate attained by the {\AA}-B scheme decreases as the minimum sampling period $\tau$ decreases but there is little improvement after $\tau<10^{-2}$ seconds. For each $\tau$, the function $\theta^{-1}(\tau \mathcal{I}^c(D))/\tau$ indeed provides a lower bound to the rate-distortion function attained by the {\AA}-B scheme. However, the gap between the two widens as $\tau$ tends to zero, and furthermore, the lower bound becomes vacuous in this limit. In the scalar case, one can show explicitly that both the DT and CT bounds give tend to $0$ for all feasible distortions. This suggests that this lower bound most informative at intermediate values of $\tau$, where $\tau$ is small enough to keep the error from diverging but not so small that the bound becomes essentially vacuous.

\subsection{Linear vector source}\label{subsec:vecnum}
There is no direct analog to \cite{astrom2002comparison} in for higher-dimensional sources. To gauge the tightness of our lower bounds for vector plants, we simulated an approach similar to that in \cite{tccdc23} to establish an upper bound. The encoding approach consists of two stages; first, the source is sampled with period $\tau$ and quantized via dithered innovations quantization (DIQ, cf. e.g. \cite{tccdc23,tccjsait}), then the (discrete) quantization is losslessly encoded into a binary codeword. The decoder performs these operations in reverse and updates its estimate with the decoded, dequantized measurement. Note that this quantizer design requires a dither signal, nominally a sequence of uniform random variables, independent of the plant process, that are assumed to be shared by the encoder and decoder. For each $\tau$, we design a quantizer based on the solution to (\ref{eq:RDF_disc_tc_vec_explicit}). Both the encoder and decoder maintain an empirical probability mass function
of the quantizer's output. A codebook maps the realizations in the model to the set of binary strings. At each timestep, the codebook is optimized to map more (empirically) likely realizations to shorter strings. As in \cite{tccdc23}, the support of the quantizer's output is countably infinite, and thus the empirical model cannot include all possible realizations. When a realization occurs that falls outside the model, the ``truncation" procedure is defined using the Elias Omega code \cite{eliasUniversal}. We provide a detailed explanation of this achievability approach in an Appendix which can be found with the simulation source code (and hyperparameter definitions) in \cite{github_appendix}. 

This encoder was simulated for a crude model of the phugoid mode of a Boeing 747-100 flying at Mach 0.5 at an altitude of 20,000 ft. The model was generated using approximations from \cite{drela2014flight} and aircraft data from \cite{etkin1995dynamics}. The simplified system is governed by
\begin{equation}\label{eq:phugoid}
    \left[\begin{matrix}
        d(\Delta \rvec{u}) \\
        d(\Delta \rvec{\theta})
    \end{matrix}\right] = \left[\begin{matrix}
        X_u/m & -g \\
        -Z_u/(mu_0) & 0
    \end{matrix}\right] \left[\begin{matrix}
        \Delta \rvec{u} \\
        \Delta \rvec{\theta}
    \end{matrix}\right] dt +\dmat{B}d\rvec{w}_t,
\end{equation}
where $\Delta \rvec{u}$ is the aircraft's forward airspeed perturbation, $\Delta \rvec{\theta}$ is its pitch perturbation, $m$ is its mass, $u_0$ is its unperturbed forward airspeed, $X_u$ and $Z_u$ are stability derivatives and $g$ is acceleration due to gravity. We simulated this system together with the aforementioned ad hoc coding/quantization approach over a horizon of 10,000 samples. We plot both the DT lower bound from Theorem \ref{thm:firstbound} and the CT bound from \ref{thm:keybounds_abbrev}.

\begin{figure}[h]
    \centering \includegraphics[width=0.48\textwidth]{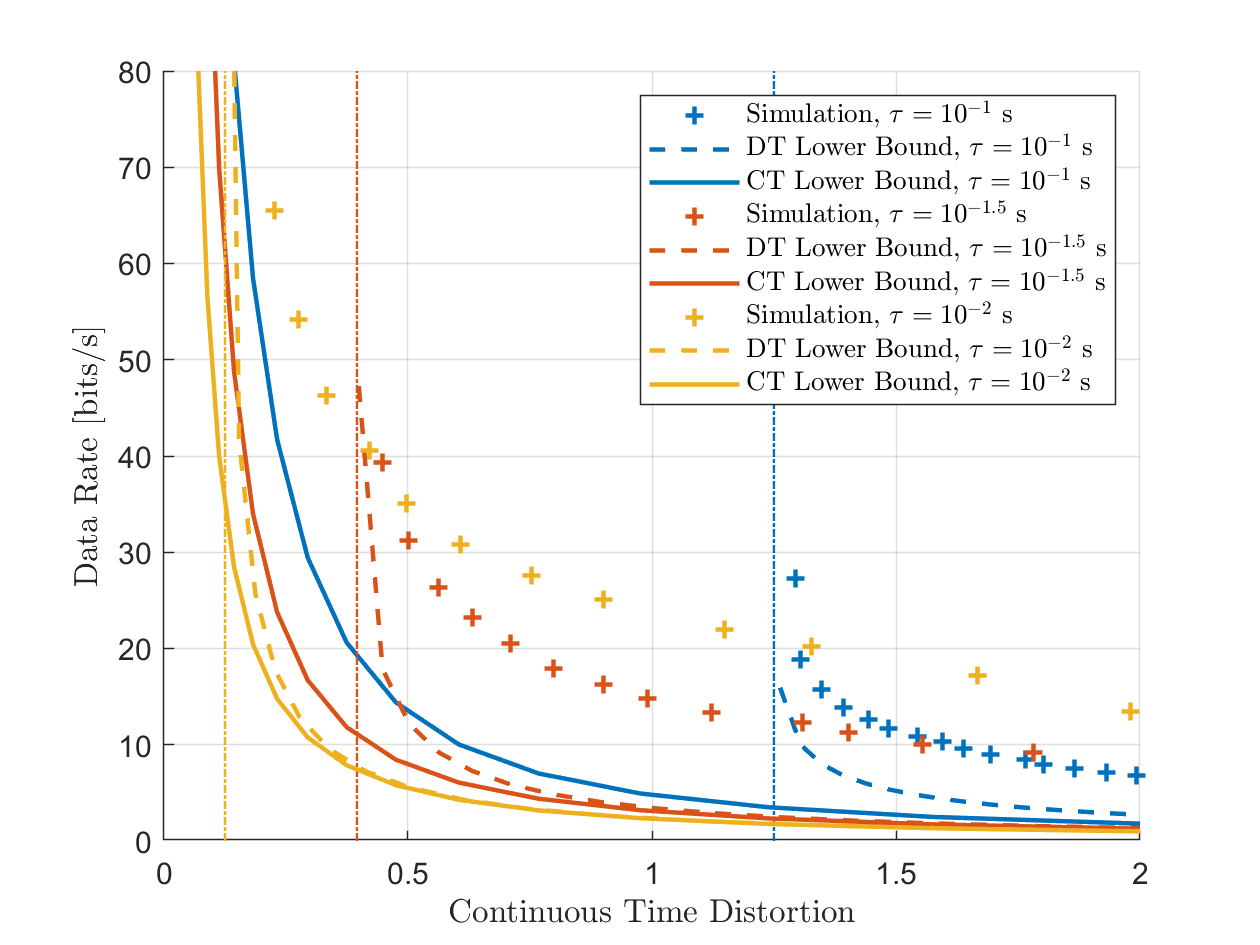}
    \caption{Comparison of the simulated encoder to the CT lower bound $\theta^{-1}(\tau \mathcal{I}^c(D))/\tau$ and DT lower bound $\theta^{-1}( R(D_{d,\tau},\overline{\dmat{Q}}_{\dmat{A},\tau},\tau))/\tau$. }
    \label{fig:b747_sim_RD}
\end{figure}

As expected, the DT lower bound is tighter than the continuous time bound. The best continuous-time distortion achievable by a discrete-time system is $\underline{D_{c}} = \overline{b_{\tau}}/\tau$. This corresponds to noiseless observations of the plant at every $k=k\tau$, and is the critical value of $D_{c}$ that makes $ D_{d,\tau } = 0$ in (\ref{eq:DcDd_vec}). We plot this critical distortion for each $\tau$ with a  vertical dotted line. One can see that the DT lower bound captures this as an asymptote; as $D_{c}$ approaches $\underline{D_{c}}$ from the right, the lower bound on bitrate tends to infinity. The continuous-time lower bound does not capture this behavior; it tracks the discrete-time bound reasonably well for distortions above the critical value, but gives a finite bitrate at the asmyptote. 

Again, there is a significant gap between the performance achieved by the simulated quantizer and source codec and the lower bound. In this case, however, the simulated sampling/compression architecture is not known to be optimal (we suspect it is far from so). Indeed, one will note that for some distortions a system with a higher sampling frequency (lower-$\tau$) gives a higher bitrate than the system that uses a lower one. To an extent, this is to be expected. Let $\rvec{q}_{t}$ denote the quantizers output in the simulated system. The analysis of the quantizer architecture in \cite{tccjsait} suggests that the quantizer's output entropy goes like $H(\rvec{q}_{t})\sim R(D,\dmat{Q},\tau)+1.26n$, where $n$ is the plant dimension. In other words, the quantizer's output entropy has a constant per sample overhead, regardless of the sampling rate.  When the $\rvec{q}_{t}$ are quantized, the bitrate goes like  $\theta^{-1}(H(\rvec{q}_{t}))/\tau$, which goes to infinity as $\tau \rightarrow 0$. The use of a lower sampling rate allows one to amortize this overhead over more time, and can reduce the bitrate (see also \cite{tanaka2016optimalblocklength}). Note that at each distortion, the best upper bound for each sampling period can be read as the lowest empirical curve whose sampling period is at  least target $\tau$. 

\section{Discussion and future work}
While we have derived novel lower bounds on the minimum bitrate required to track a multidimensional plant with an event-based encoding and constrained sampling rate, Section \ref{sec:numerical} illustrated a significant gap between our lower bound and upper bounds on the optimal bitrate. This becomes especially pronounced in the CT $\tau \rightarrow 0$ limit, where both the CT and DT lower bounds from Theorems \ref{thm:keybounds_abbrev} and \ref{thm:firstbound} become vacuous. In the scalar case (i.e. $\dmat{A}\in\mathbb{R}$ we can 
show that (\ref{eq:rdf_midef}) can be written 
\begin{equation}
\label{eq:RDF_scalar}
R(D_{d,\tau },\dmat{Q}_{\dmat{A},\tau},\tau)=\max\left\{0, \frac{1}{2}\log_2\left(\dmat{A}_{\tau}^2+\frac{\dmat{B}_{\tau}^2}{D_{d,\tau }/\dmat{Q}_{\dmat{A},\tau}}\right)\right\}.
\end{equation} Given the expression for $\mathcal{I}^{c}(D_{c})$(\ref{eq:1d_ctid}), we have the following.
\begin{lemma}\label{lemm:comparison}
For every $\dmat{A}\in\mathbb{R}$, $\dmat{A}<0$, $\dmat{B}\in\mathbb{R}$, fixed $D_{c}>0$, and $D_{d,\tau}$ given by (\ref{eq:DcDd_vec}), then $\tau = 0$ is a local minimum of the function $f(\tau) =\frac{R(D_{d,\tau },\dmat{Q}_{\dmat{A},\tau},\tau)}{\tau}$, and  $\lim_{\tau\rightarrow 0}f(\tau) = \mathcal{I}^{c}(D_{c})$.
\end{lemma} 
\begin{proof}
Note that all the arguments of $R(D_{d,\tau },\dmat{Q}_{\dmat{A},\tau},\tau)$ depend on $\tau$. Substituting for these arguments in terms of $\tau$, one can show the limit, that $\underset{\tau\rightarrow 0}{\lim}\frac{d}{d\tau}f(\tau)=0$, and $\underset{\tau\rightarrow 0}{\lim}\frac{d^2}{d\tau^2}f(\tau)>0$ using symbolic software.
\end{proof} While this enables us to verify that the DT bound from Theorem \ref{thm:firstbound} and the CT bound of Theorem \ref{thm:keybounds_abbrev} coincide for small $\tau$, unfortunately both bounds become vacuous, as we now show. \begin{theorem}
\label{thm:vacuousLB}
For $\dmat{A}$, $\dmat{B}$, $D_{c}$, and $D_{d,\tau}$ as in Lemma \ref{lemm:comparison}, we have $\lim_{\tau\rightarrow 0}\frac{1}{\tau}\theta^{-1}(R(D_{d,\tau },\dmat{Q}_{\dmat{A},\tau},\tau)) = 0.$
\end{theorem}
\begin{proof}
Applying l'H\^{o}pital's rule, we obtain $\lim_{\tau\rightarrow 0}\frac{1}{\tau}\theta^{-1}(R(D_{d,\tau },\dmat{Q}_{\dmat{A},\tau},\tau))=\lim_{y\rightarrow 0} \frac{d}{dy}\theta^{-1}(y)\cdot \lim_{\tau\rightarrow 0} \frac{1}{\tau} R(D_{d,\tau },\dmat{Q}_{\dmat{A},\tau},\tau)$.
By Lemma~\ref{lemm:comparison}, we have $\lim_{\tau\rightarrow 0} \frac{1}{\tau} R(D_{d,\tau },\dmat{Q}_{\dmat{A},\tau},\tau)=\mathcal{I}^{c}(D_{c})$. 
However, by definition of $\theta(\cdot)$, $\theta^{-1}(\cdot)$ has a vanishing gradient at the origin, i.e., $\lim_{y\rightarrow 0} \frac{d}{dy}\theta^{-1}(y)=0$.
\end{proof} Thus, our bounds become vacuous for small $\tau$ owing to the fact that $\theta^{-1}$ has a vanishing gradient. The development of tighter lower bounds is thus an opportunity for future work. 

The proof of the lower bound in Theorem \ref{thm:firstbound} follows from, essentially, from five uses of the $\le$ symbol. As we already indicated, Prop. \ref{prop:cd2dtprop} is actually an equality. The remaining lower bounds are  Lemma \ref{lemm:mutual_information_lb}'s   (\ref{eq:the_lb_we_relax}),  (\ref{eq:lb_relaxed_once}), (\ref{eq:l1pf_defmi}), and Theorem \ref{thm:firstbound}'s (\ref{eq:optim_first_simpprime}). The first of these, (\ref{eq:the_lb_we_relax}), follows from the well-established tractable lower bound on the bitrate of lossless encoding without prefix constraints (\ref{eq:noprefix}) from \cite{szpankowski2011minimum}. To date, we are not aware of a tractable bound tighter than (\ref{eq:noprefix}). We suspect that the first major ``loosening" of our bound occurs in the application of Jensen's inequality in (\ref{eq:lb_relaxed_once}). The difference between the minimum achievable encoder output entropy on the right hand side of (\ref{eq:l1pf_defmi}) and the rate-distortion function \ref{eq:rdf_midef} has been well studied in the literature \cite{tccjsait,tanaka2017lqg,kostina2019rate}, and is known to be tight. This suggests to us that (\ref{eq:l1pf_defmi}) and (\ref{eq:optim_first_simpprime}) are less promising candidates for tightening.

\bibliographystyle{IEEEtran}
\bibliography{references}

\appendix
\begin{figure*}[h]
    \centering   \includegraphics[width=0.8\textwidth]{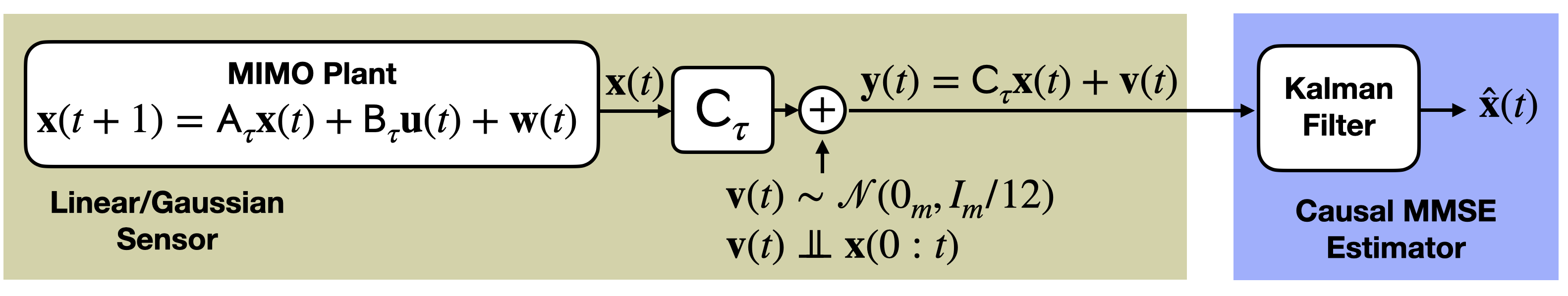}
    \caption{Denote the optimal $\dmat{P}$ and $\dmat{\Pi}$ from (\ref{eq:RDF_disc_tc_vec_explicit}) as $\dmat{P}^{*}$ and $\dmat{\Pi}^*$. It can be shown that the optimal ($\dmat{P}^{*}$, $\dmat{\Pi}^*$) have $\dmat{\Pi}^*= \dmat{A}_{\tau}\dmat{P}^{*}\dmat{A}_{\tau}^{\top}+ \dmat{B}_{\tau}\dmat{B}_{\tau}^{\top}$, and that the  minimum satisfies $R(D,\dmat{Q},\tau) = \frac{1}{2}\log_{2}(\frac{\det \dmat{\Pi}^{*}}{\det\dmat{P}^{*}})$
    \cite{tanaka2016semidefinite}. Let $\dmat{C}_{\tau}$ be any matrix such that $\dmat{C}_{\tau}^{\top}\dmat{C}_{\tau}=(\dmat{P}^{-1}-\dmat{\Pi}^{-1})/12$.
    In the above architecture, given this choice of $\dmat{C}_{\tau}$, it can be shown that the Kalman filter's posterior error covariance has $\lim_{t\rightarrow\infty}\mathbb{E}\left[(\rvec{x}_{t}-\rvec{\hat{x}}_{t})(\rvec{x}_{t}-\rvec{\hat{x}}_{t})^{\top}\right] =\dmat{P}^{*}$. The filter's prediction error covariance tends to $\dmat{\Pi}^{*}$. It can be shown that the reconstruction kernel (\ref{eq:reconkern}) induced by the above architecture attains (\ref{eq:rdf_midef}'s minimum time-average mutual information $\limsup_{K\rightarrow \infty}I(\bx(0:K-1);\hat{\bx}(0:K-1))/K = R(D,\dmat{Q},\tau)$. Thus the optimization (\ref{eq:rdf_midef}) has the interpretation of choosing, by optimization, a sensing matrix $C_{\tau}$ such that the linear-Gaussian measurement $\rvec{y}(t)=C_{\tau}\rvec{x}(t)+\rvec{v}(t)$ (with $\rvec{v}(t)\indep \rvec{x}(0:t)$) conveys the minimum information required to enable the decoder's estimate to satisfy the distortion constraint.}
    \label{fig:rdf_explanation}
    \vspace{-.5cm}
\end{figure*}
In these appendices, we describe the algorithm simulated in Section \ref{subsec:vecnum} to upper bound $\mathcal{L}^*_{\mathrm{CT}}(D_{c},\tau)$. The approach is based on that of \cite[Section IV]{tccdc23}, which defines a quantizer and prefix-free encoding algorithm for discrete-time LQG control with minimum bitrate. 

The dithered innovations quantizer design used in \cite[Section IV]{tccdc23} (defined originally in \cite{silvafirst,tanaka2016rate}) can be adapted to the discrete-time Gauss-Markov tracking problem. For each minimum sampling period $\tau$, we use a likewise adapted quantizer to discretize plant measurements sampled at the maximal rate. 
The quantizer design is discussed in Appendix \ref{app:quantizer}. The quantizations are then losslessly encoded and conveyed to the decoder/estimation center. Given the causally received quantizations, the decoder estimates the plant's state. 

What makes our approach ``event-based" is how the quantizations are encoded. The encoding of quantizations is described in Section \ref{app:en} While \cite{tccdc23} encodes quantizations with a prefix-free code, we use an non-singular encoding that is not prefix. This enables the use of the no transmission symbol $\emptyset$. In a sense, the encoding used here is simpler than that of \cite{tccdc23}. 

\appendices

\section{Description of the Quantizer Design}\label{app:quantizer}

Fig. \ref{fig:rdf_explanation} describes an interpretation of the discrete-time rate-distortion function (\ref{eq:RDF_disc_tc_vec_explicit}). This intuition will be used to define the quantizer. The intuition behind this quantizer design is discussed extensively in \cite{tanaka2016rate,tccjsait}. 

For a fixed sampling period $\tau$, we solve the rate-distortion semidefinite program (\ref{eq:RDF_disc_tc_vec_explicit}), and compute the optimal measurement matrix $\dmat{C}_{\tau}$ from Fig. \ref{fig:rdf_explanation}. The encoder and the decoder maintain almost synchronized discrete-time Kalman filters with respect to the $\tau$ sampled model. We use dithered quantization (cf. \cite{tccdc23}), and so assume that the encoder and decoder share access to a sequence of uniform random variables independent of the plant state.

We first describe the encoder. Denote the encoder's Kalman filter a priori (posterior) estimate at $t=k\tau$ as $\hat{\rvec{x}}(k|k-1)$ ($\hat{\rvec{x}}(k|k)$). Define an elementwise uniform quantizer with sensitivity $\Delta=1$ as the map from $\mathbb{R}^{m}\rightarrow\mathbb{Z}^{m}$ given by
\begin{align}\label{eq:unifquantdef}
    [Q_{1}(\dmat{x})]_{{i}} = \begin{cases}
        j\text{ if }[\dmat{x}]_{{i}}\in [j-\frac{1}{2},j+\frac{1}{2})
    \end{cases}.
\end{align} The map (\ref{eq:unifquantdef}) rounds each element of $\dmat{x}$ to the nearest integer. Define the dither sequence as $\{\rvec{d}_{k}\}$ as a countable sequence of IID uniform random vectors whose elements are mutually independent and marginally uniformly distributed on $[-.5,.5]$. We assume that $\{\rvec{d}_{k}\} \indep \rvec{x}_{t\in [0,\infty]}$.  At each discrete-time $k$, the encoder computes the quantized plant measurement
\begin{align}
    \rvec{q}(k) = Q_{1}\left(\dmat{C}_{\tau}(\rvec{x}(k)-\hat{\rvec{x}}(k|k-1))+\rvec{d}_{k}\right).
\end{align} The encoder will losslessly encode $\rvec{q}(k)$ into a codeword using a procedure specified in the next section. The decoder will be able to recover $\rvec{q}(k)$ exactly by time $k+1$. The encoder then computes the ``decoder's effective measurement" via $\rvec{m}(k) =  \rvec{q}(k)  - \rvec{d}_{k}$. It can be shown that, for $\rvec{v}(k)$ a random vector whose elements are IID uniform on $[-.5,.5]$ we have
\begin{IEEEeqnarray}{rCl}
    \rvec{m}(k)  &=& \dmat{C}_{\tau}(\rvec{x}(k)-\hat{\rvec{x}}(k|k-1))+\rvec{v}(k),
\end{IEEEeqnarray}  and $\rvec{v}(k)\indep (\rvec{x}({0:k}), \hat{\rvec{x}}(0:k|-1:k-1), \hat{\rvec{x}}(0:k-1|0:k-1))$. The encoder then computes the centered measurement $ \rvec{y}_{k}=\rvec{m}(k)+\dmat{C}_{\tau}(\rvec{x}(k)\hat{\rvec{x}}(k|k-1)$, which gives the plant measurement 
\begin{align}\label{eq:centeredemas}
     \rvec{y}(k)  = \dmat{C}_{\tau}\rvec{x}(k)+\rvec{v}(k).
\end{align} Note that  $\mathbb{E}[\rvec{v}(k)\rvec{v}^{\top}(k)] = \dmat{I}_{m}/12$. Thus, the measurement (\ref{eq:centeredemas}) is equivalent, up to second order, to the optimal measurement in Fig. \ref{fig:rdf_explanation}. The encoder then updates its Kalman filter's a priori estimate  $\hat{\rvec{x}}(k|k-1)$ using the measurement  $\rvec{y}(k)$. 

\section{Description of the lossless encoding/decoder procedure.}\label{app:en}
The support of the random variable $\rvec{q}(k)$ is countably infinite, which complicates the lossless encoding of $\rvec{q}(k)$ into binary strings. As in \cite{tccdc23}, we encode more probable realizations of $\rvec{q}(k)$ using a zero-delay lossless coding technique, and encode less likely realizations differently. This permits the use of fixed precision arithmetic and reduces computation time in encoding/decoding.  

Define $\dvec{c}\in\mathbb{N}^{{m}}_{+}$ be a vector of \textit{cutoffs}. The $\dvec{c}$ are fixed hyperparameters for our encoding algorithm. Assume arithmetic is to be performed with $p$ bit unsigned integers. We require that $ {n}=\prod_{j=1}^{{m}}([\dvec{c}]_{{j}}+1)$ has ${n} \le 2^{p}$. In practice, $n$ controls the computational complexity of the encoder, and so we generally use ${n} \ll 2^{p}$. 

To encode $\rvec{q}(k)$, we transform it into a source on $\mathbb{N}_{+}^{{m}}$ by computing the vector $\rvec{s}(k)$ via
\begin{align}\label{eq:wrap}
    [\rvec{s}(k)]_{{i}} = \begin{cases}
         2[\rvec{q}(k)]_{{i}}\text{, }&[\rvec{q}(k)]_{{i}}>0 \\ 
          -2[\rvec{q}(k)]_{{i}}+1\text{, }&[\rvec{q}(k)]_{{i}}\le 0 \\ 
    \end{cases}.
\end{align} 
Define the truncation operator $\text{trunk}_{\dvec{c}}:\mathbb{N}_{+}^{{m}}\rightarrow\mathbb{N}_{0}^{{m}}$  via
\begin{align}
   [\text{trunk}_{\dvec{c}}(\dvec{s})]_{{i}}  = \begin{cases}
         [\dvec{s}]_{{i}}\text{, }&  [\dvec{s}]_{{i}} \le [\dvec{c}]_{{i}} \\ 
          0\text{, }&\text{otherwise}\\ 
    \end{cases},
\end{align} i.e. given a vector input in the strictly positive quadrant,  we replace the vector components that exceed the cutoff with zero. 
Define the truncated ``wrapped" quantization sequence via
\begin{align}
    \overline{\rvec{s}}(k) = \text{trunk}_{\dvec{c}}(\rvec{s}(k)).
\end{align} We will say a truncation occurs at time $k$ if $\overline{\rvec{s}}(k) \neq {\rvec{s}}(k)$. Denote the sequence of symbols that were truncated $\hat{\rvec{s}}(k)$, so that $[\hat{\rvec{s}}(k)]_{{i}}$ is the  ${i}^{\mathrm{th}}$ symbol truncated from $\rvec{s}(k)$. Let $\mathcal{S} = \mathbb{N}_{0}^{m}\cap ([0,[\dvec{c}]_{0}]\times[0,[\dvec{c}]_{1}]\dots \times [0,[\dvec{c}]_{m-1}] )$. Note that the support of  $\overline{\rvec{s}}(k)$ is a subset of the countably finite set $\mathcal{S}$.  At every $t$, the encoder losslessly conveys $\overline{\rvec{s}}(k)$ to the decoder. The encoder and decoder maintain synchronized empirical probability models of the source $\overline{\rvec{s}}(k)$ based on observations of $\overline{\rvec{s}}(0:t-1)$, namely, they iteratively maintain a (not normalized) PMF estimator $   \hat{\boldsymbol{\mathbb{P}}}_{\overline{\rvec{s}}(k)}:\mathcal{S}\rightarrow [0,\min(2^p,k)]$. For $k\le 2^{p}-1$, we have
\begin{align}\label{eq:pmfest}
    \hat{\boldsymbol{\mathbb{P}}}_{\overline{\rvec{s}}(k)}(\dvec{s}) = \sum_{i=0}^{k-1} \mathbbm{1}_{\overline{\rvec{s}}(i)=\dvec{s}}. 
\end{align} Note that $ \hat{\boldsymbol{\mathbb{P}}}_{\overline{\rvec{s}}(k)}(\dvec{s})$ is a random variable.

At each $k$, both the encoder and decoder compute identical ``sortings" of the PMF $\hat{\boldsymbol{\mathbb{P}}}_{\overline{\rvec{s}}(k)}$. They compute the bijective mapping $\text{sort}_{\overline{\rvec{s}}(k)}:\mathcal{S}\rightarrow [1,|\mathcal{S}|]\cap \mathbb{N}$ such that if: $\text{sort}_{\overline{\rvec{s}}(k)}(\dvec{s}^{1}) \le \text{sort}_{\overline{\rvec{s}}(k)}(\dvec{s}^{2})$ then  $ \hat{\boldsymbol{\mathbb{P}}}_{\overline{\rvec{s}}(k)}(\dvec{s}^{1}) \ge\hat{\boldsymbol{\mathbb{P}}}_{\overline{\rvec{s}}(k)}(\dvec{s}^{2})$, e.g. for a source realization $\dvec{s}$, $\text{sort}_{\overline{\rvec{s}}(k)}(\dvec{s})$ outputs the rank of $\dvec{s}$, in probability order with respect to the estimated PMF (\ref{eq:pmfest}). Define the nonsingular codebook $\mathcal{C}_{\mathrm{ns}}:[1,|\mathcal{S}|]\cap \mathbb{N} \rightarrow \{0,1\}^{*}$ as in Table \ref{tab:nsencoding}. At time $k$, the encoder conveys $\overline{\rvec{s}}(k)$ to the decoder via transmitting the binary string $\rvec{a}^{1}(k) = \mathcal{C}_{\mathrm{ns}}(\text{sort}_{\overline{\rvec{s}}(k)}(\overline{\rvec{s}}(k)))$. In this way, realizations of $\overline{\rvec{s}}(k)$ that are more empirically probable with respect to $\hat{\boldsymbol{\mathbb{P}}}_{\overline{\rvec{s}}(k)}$ are mapped to shorter strings. Since  $\text{sort}_{\overline{\rvec{s}}(k)}$ and $\mathcal{C}_{\mathrm{ns}}$ are both one-to-one, they can be inverted and the decoder can recover $\overline{\rvec{s}}(k)$ exactly. The encoder and decoder then update their empirical probability models. If $k < 2^{p}-1$ or  
$\hat{\boldsymbol{\mathbb{P}}}_{\overline{\rvec{s}}(k)}(\rvec{s}(k)) < 2^{p}-1$ then we define
\begin{align}\label{eq:pmfest2}
    \hat{\boldsymbol{\mathbb{P}}}_{\overline{\rvec{s}}(k+1)}(\dvec{s}) = \hat{\boldsymbol{\mathbb{P}}}_{\overline{\rvec{s}}(k)}(\dvec{s})+\mathbbm{1}_{\dvec{s}(k)=s}.
\end{align} If $\hat{\boldsymbol{\mathbb{P}}}_{\overline{\rvec{s}}(k)}(\rvec{s}(k)) \ge 2^{p}-1$, then applying (\ref{eq:pmfest2}) will cause an arithmetic overflow. In this case, we perform a rescaling via 
\begin{align}
   \hat{\boldsymbol{\mathbb{P}}}_{\overline{\rvec{s}}(k+1)}(\dvec{s}) = \begin{cases}
    \lfloor{ \hat{\boldsymbol{\mathbb{P}}}_{\overline{\rvec{s}}(k)}(\dvec{s})}/{2} \rfloor+1 , &\dvec{s} = \overline{\rvec{s}}(k)\\ 
        \lfloor{ \hat{\boldsymbol{\mathbb{P}}} _{\overline{\rvec{s}}(k)}(\dvec{s})}/{2} \rfloor,& \dvec{s}\neq \overline{\rvec{s}}(k)
    \end{cases},
\end{align} where $\lfloor r \rfloor$ is the ``floor" operator. This is performed synchronously at the encoder and decoder. If a truncation occurs at time $k$, before transmitting the encoded version 

If a truncation occurs at time $k$, at time $k+1$, before conveying $\rvec{a}^{1}(k+1)$,  the encoder will losslessly convey the sequence of truncated symbols $\hat{\rvec{s}}(k)$ to the decoder. It will encode the string $\hat{\rvec{s}}(k)$ sequentially via the Elias Omega code \cite{eliasUniversal}\cite{tccdc23}. The Elias Omega code is a prefix-free encoding of the positive integers. Since at time $k+1$ the decoder knows $\overline{\rvec{s}}(k)$, it knows how many truncations occurred at time $k$ (i.e. it knows the length of $\hat{\rvec{s}}(k)$ by counting the zeros in $\overline{\rvec{s}}(k)$). In this case, at time $k+1$ the decoder can reconstruct the non-truncated ${\rvec{s}}(k)$ and  ${\rvec{q}}(k)$ exactly. In the event no truncations occur at time $k$, the decoder identifies that ${\rvec{s}}(k)=\overline{\rvec{s}}(k)$ and recovers ${\rvec{q}}(k)$ immediately. In this way, at time $k$ the decoder will have received (at least) $\rvec{q}(0:k-1)$, and can thus compute $\hat{\rvec{x}}(0:k-1)$ and $\hat{\rvec{x}}(0:k|-1:k-1)$.  

At every timestep, the decoder maintains an estimate of the plant state given the measurements it has received. Let $\hat{\rvec{x}}_{\mathrm{post}}(k)$ denote the decoder's estimate at time $k$ after receiving $\rvec{a}(k)$. This is used to compute CT distortion as in via (\ref{eq:finalSimplified}). Assume for now that the decoder receives $\rvec{a}(k)$, and that no truncation occurred at the previous timestep (e.g. $\overline{\rvec{s}}(k-1) = {\rvec{s}}(k-1)$). In this case, we will have
\begin{IEEEeqnarray}{rCl}\label{eq:decodersbest}
    \hat{\rvec{x}}_{\mathrm{post}}(k-1) &=& \hat{\rvec{x}}(k-1)\\ \dmat{A}_{\tau}\hat{\rvec{x}}_{\mathrm{post}}(k-1) &=& \hat{\rvec{x}}(k|k-1)
\end{IEEEeqnarray} where $\hat{\rvec{x}}(k-1)$ and $\hat{\rvec{x}}(k|k-1)$ correspond to the decoders estimates. Recall that $\hat{\rvec{x}}(k|k-1)$ is just the encoder's Kalman prediction from the effective measurements $\rvec{y}(0:k-1)$ in (\ref{eq:centeredemas}).  If there are no truncations at time $k$, then the decoder can immediately recover ${\rvec{q}}(k)$ and, since it can compute $\hat{\rvec{x}}(k|k-1)$ via (\ref{eq:decodersbest}), it can compute ${\rvec{y}}(k)$. It then applies the Kalman measurement update to $\dmat{A}_{\tau}\hat{\rvec{x}}_{\mathrm{post}}(k-1) = \hat{\rvec{x}}(k|k-1)$ and sets  $\hat{\rvec{x}}_{\mathrm{post}}(k) = \hat{\rvec{x}}(k)$.  If there is a truncation at time $k$, the decoder simply uses a prediction, setting  $\hat{\rvec{x}}_{\mathrm{post}}(k) = \dmat{A}_{\tau}\hat{\rvec{x}}_{\mathrm{post}}(k-1)$ (equivalently  $\hat{\rvec{x}}_{\mathrm{post}}(k)=\hat{\rvec{x}}(k|k-1)$).  

If a truncation occurred at time $k-1$, at time $k$ the the decoder first receives an encoded version of $\hat{\rvec{s}}(k-1)$. It can then recover $\rvec{q}(k-1)$. It uses this to recover $\hat{\rvec{x}}(k-1|k-1)$ and to compute $\hat{\rvec{x}}(k|k-1)$. It then receives $\overline{\rvec{s}}(k)$. If no truncations occur at time $k$, it computes $\rvec{y}(k)$ via (\ref{eq:centeredemas}) and the measurement update $\hat{\rvec{x}}(k|k)$. It sets 
$\hat{\rvec{x}}_{\mathrm{post}}(k)=\hat{\rvec{x}}(k|k)$. If a truncation occurs, it sets $\hat{\rvec{x}}_{\mathrm{post}}(k)=\hat{\rvec{x}}(k|k-1)$.
\begin{table}[h]
\centering
\begin{tabular}{|c||c|c|c|c|c|c|}
\hline 
$r$ & 1 & 2 & 3 & 4 & 5 & ... \\
\hline
$\mathcal{C}_{\mathrm{ns}}(r)$ & $\emptyset$ & $0$ & 1 & 00 & 01 & ... \\
\hline
\end{tabular}
\caption{The nonsingular encoding.}\label{tab:nsencoding}
\end{table}

\end{document}